%% file: arxiv.tex
\documentclass[runningheads,bookmarks,bookmarksdepth=2]{llncs}

\usepackage[T1]{fontenc}

\usepackage{graphicx}
\usepackage{wrapfig2}
\usepackage{hyperref}
\usepackage{amsmath,amssymb}
\usepackage{stmaryrd}
\usepackage[capitalise]{cleveref}
\usepackage{xcolor}

\usepackage{potl}
\usepackage{mathtools}
\usepackage{thmtools}
\usepackage{booktabs}
\usepackage{tabularx}
\usepackage{multirow}
\usepackage{adjustbox}
\usepackage{pifont}
\usepackage[all]{foreign}
\usepackage{subcaption}
\usepackage[inline]{enumitem}

\usepackage{forest}
\usepackage{tikz}
\usetikzlibrary{arrows,automata}
\usetikzlibrary{matrix}

\usepackage{ifthen}
\newboolean{arxiv}
\setboolean{arxiv}{true}

\DeclareMathOperator\smb{smb}
\let\stack\relax
\DeclareMathOperator\stack{stack}
\DeclareMathOperator\ctx{ctx}
\DeclareMathOperator\struct{struct}

\DeclareMathOperator\depth{depth}
\DeclarePairedDelimiter{\set}{\{}{\}}
\newcommand\Gargs{\mathcal{G}}

\newcommand{\cmark}{\ding{51}}

\newcommand\unr[1]{\llbracket #1 \rrbracket}

\DeclareMathOperator{\pc}{pc}
\DeclareMathOperator{\val}{val}
\DeclareMathOperator{\varap}{varap}
\DeclareMathOperator{\prop}{propvals}

\let\emptyset\varnothing

\newcounter{termrule}
\crefname{termrule}{Rule}{Rules}

\allowdisplaybreaks

\begin{document}

\title{SMT-based Symbolic Model-Checking for Operator Precedence Languages}

\author{%
  Michele Chiari\inst{1} \and
  Luca Geatti\inst{2} \and
  Nicola Gigante\inst{3} \and
  Matteo Pradella\inst{4}}
\authorrunning{M. Chiari et al.}

\institute{%
  TU Wien, Treitlstra\ss{}e 3, 1040 Vienna, Austria\\
  \email{michele.chiari@tuwien.ac.at} \and
  University of Udine, Italy\\
  \email{luca.geatti@uniud.it} \and
  Free University of Bozen-Bolzano, Italy\\
  \email{nicola.gigante@unibz.it} \and
  Politecnico di Milano, Italy\\
  \email{matteo.pradella@polimi.it}
}

\maketitle

\begin{abstract}
  Operator Precedence Languages (OPL) have been recently identified as
  a suitable formalism for model checking recursive procedural programs,
  thanks to their ability of modeling the program stack. OPL requirements
  can be expressed in the \emph{Precedence Oriented Temporal Logic} (\ac{POTL}),
  which features modalities to reason on the natural matching
  between function calls and returns, exceptions, and other advanced
  programming constructs that previous approaches, such as Visibly Pushdown
  Languages, cannot model effectively. Existing approaches for model
  checking of \ac{POTL} have been designed following the explicit-state,
  automata-based approach, a feature that severely limits their
  scalability.
  In this paper,
  we give the first symbolic, SMT-based approach for model
  checking \ac{POTL} properties. While previous approaches construct the
  automaton for both the \ac{POTL} formula and the model of the program, we
  encode them into a (sequence of) SMT formulas. The search of a trace of
  the model witnessing a violation of the formula is then carried out by an
  SMT-solver, in a Bounded Model Checking fashion.
  We carried out an experimental
  evaluation, which shows the effectiveness of the proposed solution.

\end{abstract}

\input{sections/1.introduction.tex}
\input{sections/2.preliminaries.tex}
\input{sections/3.tableau.tex}
\input{sections/4.encoding.tex}
\input{sections/5.evaluation.tex}
\input{sections/6.conclusions.tex}

\begin{credits}
\begin{wrapfigure}{l}{2cm}
\vspace{-7ex}
\includegraphics[width=2cm]{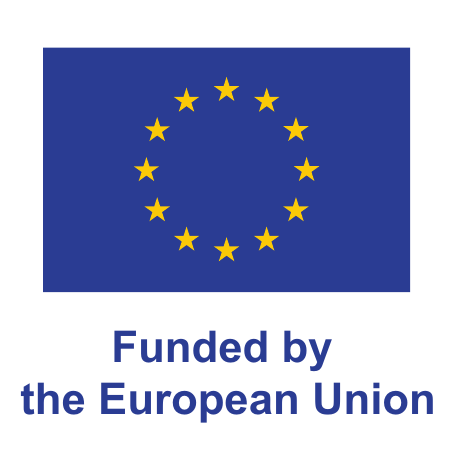}
\vspace{-2ex}
\end{wrapfigure}
\subsubsection{\ackname}
This work was partially funded by the Vienna Science and Technology Fund (WWTF)
grant [10.47379/ICT19018] (ProbInG),
and by the EU Commission in the Horizon Europe research and innovation programme
under grant agreement No.\ 101107303 (MSCA Postdoctoral Fellowship CORPORA).

\subsubsection{\discintname}
The authors have no competing interests to declare that are
relevant to the content of this article.
\end{credits}

\bibliographystyle{splncs04}
\bibliography{biblio}

\clearpage

\appendix
\input{sections/Appendix.tex}

\end{document}

%% file: sections/1.introduction.tex

\section{Introduction}
\label{sec:introduction}

Operator Precedence Languages (OPL) \cite{Floyd1963} are very promising for
software verification: as a subclass of context-free languages, they can
naturally encode the typical stack-based behavior of programs, without the
shortcomings of the better known Visibly Pushdown Languages (VPL),
originally introduced as Input-driven languages
\cite{AluMad04,jacm/AlurM09,DBLP:conf/icalp/Mehlhorn80}. In particular, the
main characteristic of VPL is the one-to-one ``matching'' between a symbol
representing a procedure call and the symbol representing its corresponding
return.  Unfortunately, this feature makes them ill-suited to model several
typical behaviors of programs that induce a many-to-one or one-to-many
matching, such as exceptions, interrupts, dynamic memory management,
transactions, and continuations.

OPL were introduced through grammars for deterministic parsing by Floyd in
1963, and were re-discovered and studied in more recent works, where
containment of VPL and closure w.r.t.\ Boolean operations were proved
\cite{CrespiMandrioli12}, together with the following characterizations:
automata-based, monadic second order logic \cite{LonatiEtAl2015},
regular-like expressions \cite{MPC20}, and syntactic congruence with
finitely many equivalence classes \cite{HKMS23}.  
OPL are also the biggest known class maintaining an important feature of Regular languages: 
first-order logic, star-free expressions, and aperiodicity define the same subclass~\cite{MPC23}.  
A temporal logic
called \ac{OPTL} was defined in \cite{ChiariMP20a}, and a subsequent
extension called \ac{POTL} (on which we focus in this work) was introduced
in \cite{ChiariMP21a}, and then proved to capture the first-order definable
fragment of OPL in \cite{ChiariMP21b}. The linear temporal logics for VPL
CaRet \cite{AlurEM04} and \ac{NWTL} \cite{lmcs/AlurABEIL08} were also
proved to be less expressive than both \ac{OPTL}~\cite{ChiariMP20a} and
\ac{POTL}~\cite{ChiariMP21b}.

\ac{POTL} contains explicit context-free modalities that interact
not only with the linear order of events representing time,
but also with the nested structure of function calls, returns, and exceptions.
For instance, consider this formula:
\[
\llglob (\lcall \land \mathrm{qs} \implies \neg (\lunext \lexc \lor \lcunext \lexc))
\]
Here $\llglob$ is the LTL globally operator,
and $\lcall$ and $\lexc$ hold respectively in positions that represent
a function call and an exception.
$\lunext \lexc$ means that the \emph{next} position is an exception
(similarly to the LTL next),
while $\lcunext \lexc$ means that a subsequent position,
which \emph{terminates} the function call in the current position, is an exception.
Thus, the formula means
``function $\mathrm{qs}$ is never terminated by an exception''
(or, equivalently, it never terminates or it always terminates with a normal return).

It is worth to note that VPL were originally proposed for automatic
verification, thanks to their nice Regular-like closure properties, but
effective Model Checking (MC) tools for them are still not publicly
available, in particular supporting logics capable of expressing
context-free specifications. This situation improved with the introduction
of POMC \cite{pomc,ChiariMP21a,ChiariMPP23}, a model checker for structured
context free languages based on \ac{POTL}, but that can be easily adapted
to the simpler structure of VPL. POMC's core consists of an explicit-state
tableau construction procedure, which yields nondeterministic automata of
size at most singly exponential in the formula's length, and is shown to be
quite effective in realistic cases in \cite{PontiggiaCP21,ChiariMPP23}.

The main shortcoming of explicit-state MC tools is the state explosion
problem, i.e. the exponential growth of the state space as the system size
and complexity increase, which makes MC infeasible for large and realistic
systems. Indeed, as reported in \cite{ChiariMPP23}, managing longer
arrays or variables encoded with a realistic number of bits was
problematic.  A classical way to address this issue is to use Symbolic
Model Checking, which is a variant of MC that represents the system and the
specification using symbolic data structures, instead of explicit
enumeration of states and transitions. One very successful symbolic
technique is Bounded Model Checking (BMC)
\cite{DBLP:journals/ac/BiereCCSZ03,ClarkeBRZ01}, where the
model is unrolled for a fixed number of steps and encoded into SAT, i.e.
Boolean Satisfiability, to leverage recent efficient SAT solvers, and later
the more general Satisfiability Modulo Theories (SMT) solvers, such as Z3
\cite{DBLP:conf/tacas/MouraB08}.

In this paper we apply BMC to \ac{POTL} by encoding its tableau into SMT,
extending the approach used in the BLACK tool~\cite{GeattiGMV24}. BLACK is
a satisfiability checker and temporal reasoning framework based on an
encoding into SAT of Reynolds' one-pass tableau system for classical linear
temporal logic~\cite{GeattiGMR21}. Currently, we consider the future
fragment of the temporal logic \ac{POTL} on finite-word semantics, but we
plan to extend the encoding to cover full \ac{POTL} and $\omega$-words.
SMT-based approaches were already introduced for verifying pushdown program
models \cite{HuangW10,KomuravelliGC16}, but only against regular
specifications. To the best of our knowledge, this is the first SMT
encoding of a context-free temporal logic, proving that BMC can be
beneficial to verification of this class of temporal logics, too.

We applied our tool to a number of realistic cases: an implementation of
the Quicksort algorithm, a banking application, and C++ implementations of
a generic stack data structure, where our approach is compared with the
original POMC. The results are very promising, as our SMT-based approach
was able to avoid POMC's exponential increase of the solving time in
several cases.

The paper is structured as follows. OPL and the logic \ac{POTL} are
introduced in Section 2.  Section 3 defines the tree-shaped tableau for
\ac{POTL}, while Section 4 presents its encoding into SMT.  Section
5 illustrates the experimental evaluation. Last, Section 6 draws the
conclusions.


%% file: sections/2.preliminaries.tex

\section{Preliminaries}
\label{sec:preliminaries}

\subsection{Operator Precedence Languages}

We assume that the reader has some familiarity with formal language theory concepts such as
context-free grammar, parsing, shift-reduce algorithm \cite{GruneJacobs:08,Harrison78}.
Operator Precedence Languages (OPL) were historically defined through their generating grammars
\cite{Floyd1963}; in this paper, we
characterize them through their automata
\cite{LonatiEtAl2015}, as they are more suitable for model checking.
Readers not familiar with OPL may refer to \cite{MP18} for more 
explanations on their basic concepts.

Let $\Sigma$ be a finite alphabet, and $\varepsilon$ the empty string.
We use a special symbol $\# \not\in \Sigma$ to mark the beginning and
the end of any string.
  An \textit{operator precedence matrix} (OPM) $M$ over $\Sigma$ is a partial function
  $(\Sigma \cup \{\#\})^2 \to \{\lessdot, \allowbreak \doteq, \allowbreak \gtrdot\}$,
  that, for each ordered pair $(a,b)$, defines the \emph{precedence relation} (PR) $M(a,b)$
  holding between $a$ and $b$. If the function is total we say that M is \emph{complete}.
  We call the pair $(\Sigma, M)$ an \emph{operator precedence alphabet}.
  Relations $\lessdot, \doteq, \gtrdot$, are respectively named
  \emph{yields precedence, equal in precedence}, and \emph{takes precedence}.
  By convention, the initial \# yields precedence, and other
  symbols take precedence on the ending \#.
  If $M(a,b) = \prf$, where $\prf \in \{\lessdot, \doteq, \gtrdot \}$,
  we write $a \pr b$.  For $u,v \in \Sigma^+$ we write $u \pr v$ if
  $u = xa$ and $v = by$ with $a \pr b$.
  The role of PR is to give structure to words:
  they can be seen as special and more concise parentheses, where
  e.g. one ``closing'' $\gtrdot$ can match more than one ``opening'' $\lessdot$.
  It is important to remark that PR are not ordering relations, 
  despite their graphical appearance.

\begin{definition}\label{def:OPA}
An  \emph{operator precedence automaton (OPA)} is a tuple
$\mathcal A = (\Sigma, \allowbreak M, \allowbreak Q, \allowbreak I, \allowbreak F, \allowbreak \delta) $ where
$(\Sigma, M)$ is an operator precedence alphabet,
$Q$ is a finite set of states,
$I \subseteq Q$ is the set of initial states,
$F \subseteq Q$ is the set of final states,
$\delta$ is a triple of transition relations
$\delta_{\mathit{shift}}\subseteq Q \times \Sigma \times Q$,
$\delta_{\mathit{push}}\subseteq Q \times \Sigma \times Q$,
and
$\delta_{\mathit{pop}}\subseteq Q \times Q \times Q$.
An OPA is deterministic iff $I$ is a singleton,
and all three components of $\delta$ are functions.
\end{definition}

To define the semantics of OPA, we set some notation.
Letters $p, q, p_i, \allowbreak q_i, \dots$ denote states in $Q$.
We use
$q_0 \va{a}{q_1}$ for $(q_0, a, q_1) \in \delta_{\mathit{push}}$,
$q_0 \vshift{a}{q_1}$ for $(q_0, a, q_1) \in \delta_{\mathit{shift}}$,
$q_0 \flush{q_2}{q_1}$  for $(q_0, q_2, q_1) \in  \delta_{\mathit{pop}}$,
and ${q_0} \ourpath{w} {q_1}$, if the automaton can read $w \in \Sigma^*$ going from $q_0$ to $q_1$.
Let $\Gamma = \Sigma \times Q$ and $\Gamma' = \Gamma \cup  \{\bot\}$
be the \textit{stack alphabet};
we denote symbols in $\Gamma'$ as $\tstack aq$ or $\bot$.
We set $\symb {\tstack aq} = a$, $\symb {\bot}=\#$, and
$\state {\tstack aq} = q$.
For a stack content $\gamma = \gamma_n \dots \gamma_1 \bot$,
with $\gamma_i \in \Gamma$, $n \geq 0$,
we set $\symb \gamma = \symb{\gamma_n}$ if $n \geq 1$, $\symb \gamma = \#$ if $n = 0$.

A \emph{configuration} of an OPA is a triple $c = \tconfig w q \gamma$,
where $w \in \Sigma^*\#$, $q \in Q$, and $\gamma \in \Gamma^*\bot$.
A \emph{computation} or \emph{run} is a finite sequence
$c_0 \transition{} c_1 \transition{} \dots \transition{} c_n$
of \emph{moves} or \emph{transitions}
$c_i \transition{} c_{i+1}$.
There are three kinds of moves, depending on the PR between the symbol
on top of the stack and the next input symbol:

\noindent {\bf push move:} if $\symb \gamma \lessdot \ a$ then
$\tconfig {ax} p  \gamma \transition{} \tconfig {x} q {\tstack   a p \gamma }$,
with $(p,a, q) \in \delta_{\mathit{push}}$;

\noindent {\bf shift move:} if $a \doteq b$ then
$\tconfig {bx} q { \tstack a p \gamma}  \transition{} \tconfig x  r { \tstack b p \gamma}$,
with $(q,b,r) \in \delta_{\mathit{shift}}$;

\noindent {\bf pop move:} if $a \gtrdot b$
then
$\tconfig {bx} q  { \tstack a p \gamma}\transition{} \tconfig {bx} r \gamma $,
with $(q, p, r) \in \delta_{\mathit{pop}}$.

Shift and pop moves are not performed when the stack contains only $\bot$.
Push moves put a new element on top of the stack consisting of the input symbol together with the current state of the OPA.
Shift moves update the top element of the stack by \textit{changing its input symbol only}.
Pop moves remove the element on top of the stack,
and update the state of the OPA according to $\delta_{\mathit{pop}}$ on the basis of the current state of the OPA and the state of the removed stack symbol.
They do not consume the input symbol, which is used only to establish the $\gtrdot$ relation, remaining available for the next move.
The OPA accepts the language
$
L(\mathcal A) = \left\{ x \in \Sigma^* \mid  \tconfig {x\#} {q_I} {\bot} \vdash ^*
\tconfig {\#} {q_F}{\bot} , \allowbreak q_I \in I, \allowbreak q_F \in F \right\}.
$

We now introduce the concept of {\em chain}, which makes the connection between OP relations and
context-free structure explicit, through brackets.
\begin{definition}\label{def:chain}
A \emph{simple chain}
$
\ochain {c_0} {c_1 c_2 \dots c_\ell} {c_{\ell+1}}
$
is a string $c_0 c_1 c_2 \dots c_\ell c_{\ell+1}$,
such that:
$c_0, \allowbreak c_{\ell+1} \in \Sigma \cup \{\#\}$,
$c_i \in \Sigma$ for every $i = 1,2, \dots \ell$ ($\ell \geq 1$),
and $c_0 \lessdot c_1 \doteq c_2 \dots c_{\ell-1} \doteq c_\ell
\gtrdot c_{\ell+1}$.
A \emph{composed chain} is a string 
$c_0 s_0 c_1 s_1 c_2  \dots c_\ell s_\ell c_{\ell+1}$, 
where
$\ochain {c_0}{c_1 c_2 \dots c_\ell}{c_{\ell+1}}$ is a simple chain, and
$s_i \in \Sigma^*$ is the empty string 
or is such that $\ochain {c_i} {s_i} {c_{i+1}}$ is a chain (simple or composed),
for every $i = 0,1, \dots, \ell$ ($\ell \geq 1$). 
Such a composed chain will be written as
$\ochain {c_0} {s_0 c_1 s_1 c_2 \dots c_\ell s_\ell} {c_{\ell+1}}$.
$c_0$ (resp.\ $c_{\ell+1}$) is called its \emph{left} (resp.\ \emph{right}) \emph{context};
all symbols between them form its \emph{body}.
\end{definition}

\begin{figure}[tb]
\centering
\begin{tabular}{p{0.3\textwidth} p{0.7\textwidth}}
\begin{minipage}{0.3\textwidth}
\[
\begin{array}{r | c c c c}
         & \lcall   & \lret   & \lhandle & \lthrow \\
\hline
\lcall   & \lessdot & \doteq  & \lessdot & \gtrdot \\
\lret    & \gtrdot  & \gtrdot & \gtrdot & \gtrdot \\
\lhandle & \lessdot & \gtrdot & \lessdot & \doteq \\
\lthrow  & \gtrdot  & \gtrdot & \gtrdot  & \gtrdot \\
\end{array}
\]
\end{minipage}
&
\begin{minipage}{0.7\textwidth}
\begin{footnotesize}
\[
\begin{array}{l | l}
1 & \# \lessdot \lcall \lessdot \lhandle \lessdot \underline{\lcall} \gtrdot \lthrow \gtrdot \lcall \doteq \lret \gtrdot \lret \gtrdot \# \\
2 & \# \lessdot \lcall \lessdot \underline{\lhandle} \doteq \underline{\lthrow} \gtrdot \lcall \doteq \lret \gtrdot \lret \gtrdot \# \\
3 & \# \lessdot \lcall \lessdot \underline{\lcall} \doteq \underline{\lret} \gtrdot \lret \gtrdot \# \\
4 & \# \lessdot \underline{\lcall} \doteq \underline{\lret} \gtrdot \# \\
5 & \# \doteq \# \\
\end{array}
\]
\end{footnotesize}
\end{minipage}
\end{tabular}
\[
\# [ \lcall [ [ \lhandle
[ \lcall ]
\lthrow ] \lcall \; \lret ] \lret ] \#
\]
\caption{OPM $M_\lcall$ (left), a string with chains shown by brackets (bottom), and its parsing steps using the OP algorithm (right).}
\label{fig:opm-mcall}
\end{figure}

A finite word $w$ over $\Sigma$ is \emph{compatible} with an OPM $M$ iff
for each pair of letters $c, d$, consecutive in $w$, $M(c,d)$ is defined and,
for each substring $x$ of $\# w \#$ that is a chain of the form $^a[y]^b$,
$M(a, b)$ is defined.

Chains can be identified through the traditional operator precedence parsing algorithm.
We apply it to the sample word
\(w_\mathit{ex} =
 \lcall \ \allowbreak
 \lhandle \ \allowbreak
 \lcall \ \allowbreak
 \lthrow \ \allowbreak
 \lcall \ \allowbreak
 \lret \ \allowbreak
 \lret
\),
which is compatible with $M_\lcall$.
First, write all precedence relations between consecutive characters,
according to $M_{\lcall}$.
Then, recognize all innermost patterns of the form
$a \lessdot c \doteq \dots \doteq c \gtrdot b$ as simple chains, and remove their bodies.
Then, write the precedence relations between the left and right contexts of the removed body,
$a$ and $b$, and iterate this process until only \#\# remains.
This procedure is applied to $w_\mathit{ex}$ and illustrated in Fig.~\ref{fig:opm-mcall} (right).
The chain body removed in each step is underlined.
In step 1 we recognize the simple chain $\ochain{\lhandle}{\underline{\lcall}}{\lthrow}$, which can be removed.
In the next steps we recognize as chains first 
$\ochain {\lcall} {\lhandle \, \lthrow} {\lcall}$, then
$\ochain {\lcall} {\lcall \, \lret} {\lret}$, and last
$\ochain {\#} {\lcall \, \lret} {\#}$.
Fig.~\ref{fig:opm-mcall} (bottom) reports the chain structure of $w_\mathit{ex}$.

Let $\mathcal A$ be an OPA.
We call a \emph{support} for the simple chain
$\ochain {c_0} {c_1 c_2 \dots c_\ell} {c_{\ell+1}}$
any path in $\mathcal A$ of the form
$q_0
\va{c_1}{q_1}
\vshift{}{}
\dots
\vshift{}q_{\ell-1}
\vshift{c_{\ell}}{q_\ell}
\flush{q_0} {q_{\ell+1}}$.
The label of the last (and only) pop is exactly $q_0$, i.e.\ the first state of the path;
this pop is executed because of relation $c_\ell \gtrdot c_{\ell+1}$.
We call a \emph{support for the composed chain}
$\ochain {c_0} {s_0 c_1 s_1 c_2 \dots c_\ell s_\ell} {c_{\ell+1}}$
any path in $\mathcal A$ of the form
\(
q_0
\ourpath{s_0}{q'_0}
\va{c_1}{q_1}
\ourpath{s_1}{q'_1}
\vshift{c_2}{}
\dots
\vshift{c_\ell} {q_\ell}
\ourpath{s_\ell}{q'_\ell}
\flush{q'_0}{q_{\ell+1}}
\)
where, for every $i = 0, 1, \dots, \ell$:
if $s_i \neq \epsilon$, then $q_i \ourpath{s_i}{q'_i} $
is a support for the chain $\ochain {c_i} {s_i} {c_{i+1}}$, else $q'_i = q_i$.

Chains fully determine the parsing structure of any
OPA over $(\Sigma, M)$. If the OPA performs the computation
$
\langle sb, q_i, [a, q_j] \gamma \rangle \vdash^*
\langle b,  q_k, \gamma \rangle
$,
then $\ochain asb$
is necessarily a chain over $(\Sigma, \allowbreak M)$, and there exists a support
like the one above with $s = s_0 c_1 \dots c_\ell s_\ell$ and $q_{\ell+1} = q_k$.
This corresponds to the parsing of the string $s_0 c_1 \dots c_\ell s_{\ell}$ within the
contexts $a$,$b$, which contains all information needed
to build the subtree whose frontier is that string.


In \cite{CrespiMandrioli12} it is proved that Visibly Pushdown Languages (VPL) \cite{AluMad04} are strictly included in OPL.
In VPL the input alphabet is partitioned into three disjoint sets,
namely of \emph{call} ($\Sigma_c$), {\em return} ($\Sigma_r$), and {\em internal} ($\Sigma_i$) symbols,
where {\em calls} and {\em returns} respectively play the role of open and closed parentheses.
Intuitively, the string structure determined by these alphabets can be represented through an OPM as follows:
$a \lessdot b$, for any $a \in \Sigma_c$,
$b \in \Sigma_c \cup \Sigma_i$;
$a \doteq b$, for any $a \in \Sigma_c$, $b \in \Sigma_r$;
$a \gtrdot b$, for all the other cases.
On the other hand, the OPM that we use in this paper cannot be expressed in VPL, because the typical
behavior of exceptions cannot be modeled with the limited one-to-one structure of calls and returns.

To sum up, given an OP alphabet, the OPM $M$ assigns a unique structure
to any compatible string in $\Sigma^*$;
unlike VPL, such a structure is
not visible in the string, and must be built by means of a non-trivial parsing
algorithm.  An OPA defined on the
OP alphabet selects an appropriate subset within the
``universe'' of strings compatible with $M$.

\subsection{Precedence Oriented Temporal Logic}

\ac{POTL} is a propositional linear-time temporal logic featuring
context-free modalities based on OPL.  Here we are only interested in its
future fragment, \ac{POTL_f} (the letter ``f'' stands for
``future''), with the addition of \emph{weak} operators, which are needed
for our tableau. In this paper, we focus on the finite words semantics for
\ac{POTL_f}.

We fix a finite set of atomic propositions $AP$.
\ac{POTL_f} semantics are based on OP words,
which are tuples $(U, <, M_{AP}, P)$, where $U = \{0, \dots, n\}$, $n \in \mathbb{N}$,
is a finite set of word positions, $<$ a linear order on them,
$M_{AP}$ an OPM on $\powset{AP}$,
and $P : U \rightarrow \powset{U}$ a labeling function,
with $0, n \in P(\#)$.
From $M_{AP}$ follows the \emph{chain relation} $\chain \subseteq U^2$,
such that $\chain(i,j)$ holds iff $i$ and $j$ are resp.\ the left and right
contexts of a chain.
We only define the OPM on propositions in \textbf{bold}, called \emph{structural},
and assume that only one of them holds in each position.
If $\mathbf{l}_1 \sim \mathbf{l}_2$ for any PR $\sim$ and $i \in P(\mathbf{l}_1)$
and $j \in P(\mathbf{l}_2)$, we write $i \sim j$.

\begin{figure}
\centering
\begin{tikzpicture}
\matrix (m) [matrix of math nodes, column sep=-4, row sep=-4]
{
  \#
  & \color{blue} \lessdot & \lcall
  & \color{blue} \lessdot & \lhandle
  & \color{blue} \lessdot & \lcall
  & \color{blue} \lessdot & \lcall
  & \color{blue} \lessdot & \lcall
  & \color{purple} \gtrdot & \lexc
  & \color{purple} \gtrdot & \lcall
  & \color{orange} \doteq & \lret
  & \color{purple} \gtrdot & \lret
  & \color{purple} \gtrdot & \# \\
  0
  & & 1
  & & 2
  & & 3
  & & 4
  & & 5
  & & 6
  & & 7
  & & 8
  & & 9
  & & 10 \\
};
\draw[orange] (m-1-5) to [out=30, in=150] node[shift={(-10pt,-2pt)}]{$\doteq$} (m-1-13);
\draw[purple] (m-1-7) to [out=30, in=150] node[yshift=-2pt]{$\gtrdot$} (m-1-13);
\draw[purple] (m-1-9) to [out=30, in=150] node[yshift=-2pt]{$\gtrdot$} (m-1-13);
\draw[blue] (m-1-3) to [out=30, in=150] node[yshift=-2pt]{$\lessdot$} (m-1-15);
\draw[orange] (m-1-3) to [out=30, in=150] node[shift={(10pt,-3pt)}]{$\doteq$} (m-1-19);
\end{tikzpicture}
\caption{An example OP word, with the $\chain$ relation depicted by arrows, and PRs.
  First, a procedure is called (pos.~1),
  which installs an exception handler in pos.~2.
  Then, another function throws an exception,
  which is caught by the handler.
  Another function is called and returns and, finally, the initial one also returns.}
\label{fig:potl-example-word}
\end{figure}
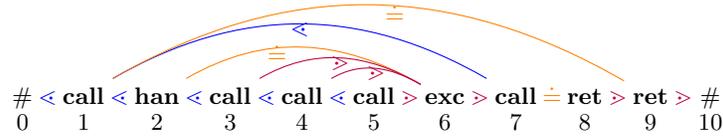

\ac{POTL_f} offers next and until operators based on two different kinds of paths,
which we define below, after fixing an OP word $w$.
\begin{definition}
\label{def:summary}
The \emph{downward summary path (DSP)} between positions $i$ and $j$, denoted $\pi_\chain^d(w, i, j)$,
is a set of positions $i = i_1 < i_2 < \dots < i_n = j$ such that, for each $1 \leq p < n$,
\[
i_{p+1} =
\begin{cases}
  k & \text{if $k = \max\{ h \mid h \leq j \land \chain(i_p,h) \land (i_p \lessdot h \lor i_p \doteq h)\}$ exists;} \\
  i_p + 1 & \text{otherwise, if $i_p \lessdot (i_p + 1)$ or $i_p \doteq (i_p + 1)$.}
\end{cases}
\]
We write $\pi_\chain^d(w, i, j) = \emptyset$ if no such path exists.
The definition for $\pi_\chain^u(w, i, j)$ is obtained by substituting $\gtrdot$ for $\lessdot$.
\end{definition}
DSPs can either go downward in the nesting structure of the $\chain$ relation
by following the linear order, or skip whole chain bodies by following the $\chain$ relation.
What this means depends on the OPM:
with $M_\lcall$, until operators on DSPs express properties local to a function invocation,
including children calls.
Their upward counterparts, instead, go from inner functions towards parent invocations.
For instance, in Fig.~\ref{fig:potl-example-word} we have
$\pi_\chain^d(w, 1, 6) = \set{1,5,6}$, and
$\pi_\chain^u(w, 2, 7) = \set{2,4,5,6,7}$.

\begin{definition}
The \emph{downward hierarchical path} between positions $i$ and $j$,
denoted $\pi_H^d(w, i, j)$, is a sequence of positions
$i = i_1 < i_2 < \dots < i_n = j$ such that there exists $h > j$ such that
for each $1 \leq p \leq n$ we have $\chain(i_p,h)$ and $i_p \gtrdot h$,
and for each $1 \leq q < n$ there is no position $k$
such that $i_q < k < i_{q+1}$ and $\chain(k,h)$.

The upward hierarchical path $\pi_H^u(w, i, j)$ is defined similarly,
except $h < j$ and for all $1 \leq p \leq n$ we have $\chain(h, i_p)$ and $h \lessdot i_p$.

We write $\pi_H^d(w, i, j) = \emptyset$ or $\pi_H^u(w, i, j) = \emptyset$
if no such path exists.
\end{definition}
Hierarchical paths range between multiple positions in the $\chain$ relation with the same one.
With $M_\lcall$, this means functions terminated by the same exception.
For instance, in Fig.~\ref{fig:potl-example-word} we have
$\pi_H^d(w, 3, 4) = \set{3,4}$.

Let $\mathrm{a} \in AP$, and $t \in \{d, u\}$;
the syntax of \ac{POTL_f} is the following:
\begin{align*}
    \varphi \coloneqq \mathrm{a} &
    \mid \neg \varphi
    \mid \varphi \lor \varphi
    \mid \lnextsup{t} \varphi
    \mid \lwnextsup{t} \varphi
    \mid \lcnext{t} \varphi
    \mid \lwcnext{t} \varphi
    \mid \lguntil{t}{\chi}{\varphi}{\varphi}
    \mid \lcrelease{t}{\varphi}{\varphi} \\
    & \mid \lhnext{t} \varphi
    \mid \lwhnext{t} \varphi
    \mid \lguntil{t}{H}{\varphi}{\varphi}
    \mid \lhrelease{t}{\varphi}{\varphi}
\end{align*}
The truth of \ac{POTL_f} formulas is defined w.r.t.\ a single word position.
Let $w$ be a finite \ac{OP} word, and $\mathrm{a} \in AP$;
we set $\sim^d = \mathord{\lessdot}$ and $\sim^u = \mathord{\gtrdot}$.
Then, for any position $i \in U$ of $w$ and $t \in \set{d, u}$:
\begin{enumerate}
  \item $(w, i) \models \mathrm{a}$ iff $i \in P(\mathrm{a})$;
  \item $(w, i) \models \neg\varphi$ iff $(w,i)\not\models\varphi$;
  \item $(w, i) \models \varphi_1\lor\varphi_2$ iff $(w,i)\models\varphi_1$ or
    $(w,i)\models\varphi_2$;
  \item $(w,i) \models \lnextsup{t} \varphi$ iff $i<|w|-1$, 
    $(w,i+1) \models \varphi$ and $i \sim^t (i+1)$ or $i \doteq (i+1)$;
  \item $(w,i) \models \lwnextsup{t} \varphi$ iff
    $i=|w|-1$ and ($i \sim^t (i+1)$ or $i \doteq (i+1)$)
    implies $(w,i+1) \models \varphi$;
  \item $(w,i) \models \lcnext{t} \varphi$
    iff $\exists j > i$ such that $\chain(i,j)$,
    $i \sim^t j$ or $i \doteq j$, and $(w,j) \models \varphi$;
  \item $(w,i) \models \lwcnext{t} \varphi$
    iff $\forall j > i$ such that $\chain(i,j)$ and
    ($i \sim^t j$ or $i \doteq j$), we have $(w,j) \models \varphi$;
  \item $(w,i) \models \lguntil{t}{\chi}{\varphi_1}{\varphi_2}$ iff
    $\exists j \geq i$ such that $\pi_\chain^t(w, i, j) \neq \emptyset$,
    $(w, j) \models \varphi_2$ and $\forall j' < j$ in $\pi_\chain^t(w, i, j)$
    we have $(w, j) \models \varphi_1$;
  \item $(w,i) \models \lcrelease{t}{\varphi_1}{\varphi_2}$ iff
    $\forall j \geq i$ such that $\pi_\chain^t(w, i, j) \neq \emptyset$
    we have either $(w, j') \models \varphi_2$ for all $j' \in \pi_\chain^t(w, i, j)$,
    or $\exists k \in \pi_\chain^t(w, i, j)$ such that $(w, k) \models \varphi_1$
    and $\forall j' \leq k$ in $\pi_\chain^t(w, i, j)$ we have $(w, j) \models \varphi_2$;
  \item $(w,i) \models \lhunext \varphi$ iff
    there exist a position $h < i$ s.t.\ $\chain(h,i)$ and $h \lessdot i$
    and a position $j = \min\{ k \mid i < k \land \chain(h,k) \land h \lessdot k \}$
    and $(w,j) \models \varphi$;
  \item $(w,i) \models \lwhunext \varphi$ iff
    the existence of a position $h < i$ s.t.\ $\chain(h,i)$ and $h \lessdot i$
    and a position $j = \min\{ k \mid i < k \land \chain(h,k) \land h \lessdot k \}$
    implies $(w,j) \models \varphi$;
  \item $(w,i) \models \lhdnext \varphi$ iff
    there exist a position $h > i$ s.t.\ $\chain(i,h)$ and $i \gtrdot h$
    and a position $j = \min\{ k \mid i < k \land \chain(k,h) \land k \gtrdot h \}$
    and $(w,j) \models \varphi$;
  \item $(w,i) \models \lwhdnext \varphi$ iff
    the existence of a position $h > i$ s.t.\ $\chain(i,h)$ and $i \gtrdot h$
    and a position $j = \min\{ k \mid i < k \land \chain(k,h) \land k \gtrdot h \}$
    implies $(w,j) \models \varphi$;
  \item $(w,i) \models \lhuntil{t}{\varphi_1}{\varphi_2}$ iff
    $\exists j \geq i$ such that $\pi_H^t(w, i, j) \neq \emptyset$,
    $(w, j) \models \varphi_2$ and $\forall j' < j$ in $\pi_H^t(w, i, j)$
    we have $(w, j) \models \varphi_1$;
  \item $(w,i) \models \lhrelease{t}{\varphi_1}{\varphi_2}$ iff
    $\forall j \geq i$ such that $\pi_H^t(w, i, j) \neq \emptyset$
    we have either $(w, j') \models \varphi_2$ for all $j' \in \pi_\chain^t(w, i, j)$,
    or $\exists k \in \pi_H^t(w, i, j)$ such that $(w, k) \models \varphi_1$
    and $\forall j' \leq k$ in $\pi_\chain^t(w, i, j)$ we have $(w, j) \models \varphi_2$.
\end{enumerate}
We additionally employ $\land$ and $\implies$ with the usual semantics.

For instance, formula $\lcduntil{\top}{p}$ evaluated in a function $\lcall$
means that $p$ holds somewhere between the call and its matched return (or exception);
formula $\lcunext p$, evaluated in a $\lcall$,
means that $p$ will hold when it returns (this can be used to check post-conditions or,
if $p = \lexc$, to assert that the function is terminated by an exception).
Formula $\lhduntil{\top}{p}$, when evaluated in a $\lcall$ terminated by an exception,
means that $p$ holds in one of the $\lcall$s already terminated by the same exception.
For a more in-depth presentation of \ac{POTL}, we refer the reader to \cite{ChiariMP21b}.


%% file: sections/3.tableau.tex

\section{A tree-shaped tableau for \texorpdfstring{\ac{POTL_f}}{POTL}}
\label{sec:tableau}

In this section, we describe our tableau system for \ac{POTL_f}, that will form the
core of our bounded model checking procedure. Let $\Sigma$ be a set of structural propositions,
$(\Sigma,M)$ an OP alphabet, AP a set of atomic propositions, and $\varphi$ a formula over
$\Sigma\cup AP$. Given $\Gamma \subseteq \clos{\varphi}$, if $\Gamma \cap \Sigma
= \{a\}$, then we define $\struct(\Gamma) = a$. Moreover, for $\Gamma, \Gamma'
\subseteq \clos{\varphi}$ and ${\sim} \in \{\lessdot, \doteq, \gtrdot\}$, we
write $\Gamma \sim \Gamma'$ meaning $\struct(\Gamma) \sim \struct(\Gamma')$.

A tableau for $\varphi$ is a tree built on top of a set of nodes $N$. Each node
$u\in N$ has four labels: $\Gamma(u)\subseteq \clos{\phi}$, $\smb(u)\in\Sigma$,
$\stack(u)\in N\cup\{\bot\}$, $\ctx(u)\in N\cup\{\bot\}$. Each node $u$ is a \emph{push}, \emph{shift}, or \emph{pop} node if, respectively, $\smb(u)\lessdot \Gamma(u)$, $\smb(u)\doteq \Gamma(u)$, or $\smb(u)\gtrdot \Gamma(u)$.

The tableau is built from $\varphi$ starting from the root $u_0$ which is
labelled as $\Gamma(u_0)=\{\varphi\}$, $\smb(u_0)=\#$, $\stack(u_0)=\bot$,
$\ctx(u_0)=\bot$. The tree is built by applying a set of \emph{rules} to each
leaf. Each rule may add new children nodes to the given leaf, while others may
\emph{accept} or \emph{reject} the leaf. The construction continues until every
leaf has been either accepted or rejected. The tableau rules can be divided into
\emph{expansion}, \emph{termination}, \emph{step}, and \emph{guess} rules. 

To each leaf of the tree, at first \emph{expansion rules}
are applied, which are summarised in \cref{table:expansion}. Each rule works as
follows. If the formula $\psi$ in the leftmost column belongs to $\Gamma(u)$,
then for each $i\in\set{1,2,3}$ for which $\Gamma_i$ is given in
\cref{table:expansion}, a child $u_i$ is added to $u$, whose labels are
identical to $u$ excepting that
$\Gamma(u_i)=(\Gamma(u)\setminus\set{\psi})\cup\Gamma_i$. If multiple rules can
be applied, the order in which they are applied does not matter. 

When no expansion rules are applicable to a leaf $u$, and
$\Gamma(u)\cap(\Sigma\cup\set{\#})=\emptyset$, then one child $u_a$, for each
$a\in\Sigma\cup\set{\#}$, is added to $u$ whose labels are the same as $u$
except
that $\Gamma(u_a)=\Gamma(u)\cup\set{a}$.

\begin{table}[t]
  \centering
  \caption{Expansion rules, where $t\in\set{u,d}$.}
  \label{table:expansion}
  \renewcommand\arraystretch{1.2}%
  \begin{tabular}{c @{\hspace{1em}} c c c}\toprule
    $\psi\in\Gamma(u)$ & $\Gamma_1$ & $\Gamma_2$ & $\Gamma_3$ \\\midrule
    $\alpha \land \beta$ & $\set{\alpha,\beta}$\\
    \midrule
    $\alpha \lor \beta$ & $\set{\alpha}$ & $\set{\beta}$\\
    $\lhuuntil{\alpha}{\beta}$  &
    $\set{\alpha,\lhnext{u} (\lhuuntil{\alpha}{\beta})}$ &
    $\set{\beta}$\rlap{~(only if condition~1 holds)}\\
    $\lhduntil{\alpha}{\beta}$  &
    $\set{\alpha,\lhnext{d} (\lhduntil{\alpha}{\beta})}$ &
    $\set{\beta}$\rlap{~(only if condition~2 holds)}\\
    $\lcrelease{t}{\alpha}{\beta}$ &
    $\{\alpha, \beta\}$ &
    $\{\beta, \lwnextsup{t} (\lcrelease{t}{\alpha}{\beta}), \lwcnext{t} 
      (\lcrelease{t}{\alpha}{\beta})\}$\\
    $\lhurelease{\alpha}{\beta}$ & $\set{\alpha,\beta}$ & 
    $\{\beta, \lwhunext(\lhurelease{\alpha}{\beta})\}$\\
    \midrule
    $\lguntil{t}{\chi}{\alpha}{\beta}$ &
    $\{\beta\}$ & 
    $\{\alpha, \allowbreak \lnextsup{t} (\lguntil{t}{\chi}{\alpha}{\beta})\}$ &
    $\{\alpha, \allowbreak \lcnext{t} (\lguntil{t}{\chi}{\alpha}{\beta})\}$\\
    $\lhdrelease{\alpha}{\beta}$ &
    $\emptyset$ & $\set{\alpha,\beta}$ &
    $\set{\beta, \lwhunext(\lhurelease{\alpha}{\beta})}$\\[-2ex]
    & \multicolumn{3}{c}{\upbracefill}\\
    & \multicolumn{3}{c}{(only if condition~2 holds)}\\
    \bottomrule\\[-2.5ex]
    condition~1: & \multicolumn{3}{l}{the closest step ancestor of $u$
    is a \emph{pop} node $u_p$}\\
    & \multicolumn{3}{l}{such that $\Gamma(\ctx(u_p)) \lessdot \Gamma(u_p)$}\\
    condition~2: & \multicolumn{3}{l}{the closest step ancestor of $u$
    is a \emph{push} or \emph{shift} node}\\
    \bottomrule
  \end{tabular}
\end{table}

When no expansion rules are applicable to a leaf $u$ and
$\Gamma(u)\cup\Sigma\ne\emptyset$, $u$ is called a \emph{step} node. In this
case, \emph{termination} rules are checked to decide whether the leaf can be
either rejected or accepted. Rejecting rules are described in
\cref{table:rejecting}. Most rules depend on the type of the leaf node $u$ where
they are applied (\ie it being a push, pop, or shift node), and the type of the
closest step ancestor $u_s$ of $u$. The rule in a given row of the table fires
when $u$ and $u_s$ are of the stated type (if any) and where the condition in
the last column is met. In this case, $u$ is rejected. We need to set up the
following terminology in order to understand some of those rules.

\begin{definition}[Fulfillment of a chain next operator]
  \label{def:xnext-fulfill}
  A $\lcdnext \alpha$ operator is said to be \emph{fulfilled} in a node $u$
  iff $\lcdnext \alpha \in \Gamma(u)$,
  and there exists a pop node descendant $u_p$ such that $\ctx(u_p) = u$ and:
  \begin{enumerate}
  \item $\Gamma(u) \lessdot \Gamma(u_p)$ or $\Gamma(u) \doteq \Gamma(u_p)$, and
  \item $\alpha \in \Gamma(u_s)$, where $u_s$ is the closest push or shift   node descending from $u_p$.
  \end{enumerate}
  Replace $\lcdnext$ with $\lcunext$ and $\lessdot$ with $\gtrdot$ for the upward case.
  \end{definition}

\begin{definition}[Pending node]
  \label{def:pending}
  A node $u$ is \emph{pending} iff either:
  \begin{enumerate}
    \item $u$ is a push node and no pop node $u_p$ exists such that $\stack(u_p) = u$, or
    \item $u$ is a shift node and no pop node $u_p$ exists such that $\stack(u_p) = \stack(u)$.
  \end{enumerate}
\end{definition}

\begin{definition}[Equivalent nodes]
  \label{def:equivalent}
  Two nodes $u$ and $u'$ belonging to the same branch are said to be \emph{equivalent} if the following hold:\\
  \begin{tabular}{p{.4\linewidth} p{.6\linewidth}}
  \begin{enumerate}
  \item $\Gamma(u) = \Gamma(u')$;
  \item $\smb(u) = \smb(u')$;
  \end{enumerate}
  &
  \begin{enumerate}[start=3]
  \item $\Gamma(\stack(u)) = \Gamma(\stack(u'))$; and
  \item $\Gamma(\ctx(u)) = \Gamma(\ctx(u'))$.
  \end{enumerate}
  \end{tabular}
\end{definition}

\newcommand\newterm[1]{%
  \refstepcounter{termrule}\arabic{termrule}.%
  \label{#1}%
}

\begin{table}[p]
  \caption{Rejecting termination rules.}
  \label{table:rejecting}
  \centering
  \renewcommand\arraystretch{1.5}%
  \begin{tabularx}{\linewidth}{r @{\hspace{1em}} l @{\hspace{1em}} l @{\hspace{1em}} X }\toprule
    $n^{\circ}$ & type of $u$ & type of $u_s{}^1$ & condition \\\midrule
    \newterm{rule:contradiction} 
      & & & $\set{p,\neg p}\subseteq\Gamma(u)$\\
    \newterm{rule:conflict} 
      & & & $|\Gamma(u)\cap\Sigma|>1$\\
    \newterm{rule:end} 
      & & & $\set{\psi,\#}\subseteq\Gamma(u)$ and $\psi$ is strong${}^2$\\
    \newterm{rule:check:pnext} 
      & & push/shift & $\Gamma(u_s)\gtrdot\Gamma(u)$ and some 
      $\ldnext \alpha \in \Gamma(u_s)$${}^1$ \\
      & & push/shift & $\Gamma(u_s) \lessdot \Gamma(u)$ and some 
      $\lunext \alpha \in \Gamma(u_s)$${}^1$ \\
    \newterm{rule:check:wpnext} 
      & & push/shift & 
      $\Gamma(u_s) \lessdot \Gamma(u)$
      or $\Gamma(u_s) \doteq \Gamma(u)$,\newline
      and some $\lwdnext \alpha \in \Gamma(u_s)$,
      but $\alpha \not\in \Gamma(u)$\\
    & & push/shift & 
      $\Gamma(u_s) \gtrdot \Gamma(u)$
      or $\Gamma(u_s) \doteq \Gamma(u)$,\newline
      and some $\lwunext \alpha \in \Gamma(u_s)$,
      but $\alpha \not\in \Gamma(u)$\\
    \newterm{rule:check:xnext} & pop & & 
      $\lcnext{t} \alpha$ is \emph{not} fulfilled in $u'$,\newline
      for some $u'\in G$ such that $\lcnext{t} \alpha \in \Gamma(u')$${}^3$,
      for $t \in \set{d, u}$\\
    \newterm{rule:check:wxnext} & push & pop &
      $\lwcdnext \alpha \in \ctx(u_s)$ and $\alpha \not\in \Gamma(u)$\\
      & shift & pop &
      $\lwcnext{t} \alpha \in \ctx(u_s)$ and $\alpha \not\in \Gamma(u)$, for $t \in \set{d, u}$\\
      & pop & pop &
      $\lwcunext \alpha \in \ctx(u_s)$ and $\alpha \not\in \Gamma(u)$\\
    \newterm{rule:check:hnextu} & pop & & 
      $\lhunext \alpha \in \Gamma(\stack(u))$ and 
      $\Gamma(\ctx(u)) \mathrel{\not\kern-2.5pt\lessdot} \Gamma(u)$\\
      & push & pop & $\lhunext \alpha \in \Gamma(\stack(u_s))$ and $\alpha \not\in \Gamma(u)$\\
      & push & push/shift & $\lhunext \alpha \in \Gamma(u)$\\
      & shift & & $\lhunext \alpha \in \Gamma(u)$\\
    \newterm{rule:check:whnextu} & push & pop &
      $\lwhunext \alpha \in \Gamma(\stack(u_s))$, $\stack(u_s)$ is a push node,
      the closest step ancestor of $\stack(u_s)$ is a pop node, and $\alpha \not\in \Gamma(u)$\\
    \newterm{rule:check:hnextd} & pop & & 
      $\lhdnext \alpha \in \Gamma(\ctx(u))$ and
      $\smb(\stack(u)) \doteq \Gamma(u)$ \\
    & pop & push/shift & $\lhdnext \alpha \in \Gamma(\ctx(u))$ \\
    & pop & pop & $\lhdnext \alpha \in \Gamma(\ctx(u))$ and
      $\alpha \not\in \Gamma(\ctx(u_s))$${}^1$\\
    & pop/shift & pop/shift & $\lhdnext \alpha \in \Gamma(u_s)$${}^1$\\
    \newterm{rule:check:whnextd} & pop & pop & 
      $\lwhdnext \alpha \in \Gamma(\ctx(u))$,\newline 
      $\smb(\stack(u)) \gtrdot \Gamma(u)$,
      and $\alpha \not\in \Gamma(\ctx(u_s))$\\
    \newterm{rule:check:huntild} & pop/shift & push/shift & 
      $\lhduntil{\alpha}{\beta} \in \Gamma(u_s)$\\
    & push/shift & pop & 
      $\lhdrelease{\alpha}{\beta}$ appears in one of the nodes between 
      $\ctx(u_s)$\newline
      and the closest step ancestor of $u_s$ (exclusive)\\
    & pop & pop & 
      $\alpha, \lwhdnext (\lhdrelease{\alpha}{\beta}) \not\in 
      \Gamma(\ctx(u_s))$,
      $\alpha, \beta \not\in \Gamma(\ctx(u_s))$, and\newline
      $\lhdrelease{\alpha}{\beta}$ appears in one of the nodes between 
      $\ctx(u_s)$\newline
      and the closest step ancestor of $u_s$ (exclusive)\\
    \newterm{rule:prune} & push/shift & & 
      there is a \emph{pending} ancestor $u_i$ of $u$ \emph{equivalent} to 
      $u$${}^4$\\
    \bottomrule
    \multicolumn{4}{l}{%
      ${}^1$~$u_s$ is the closest step ancestor of $u$
    }\\[-1.5ex]
    \multicolumn{4}{l}{%
      ${}^2$~$\psi$ is strong if it is a positive literal or a strong tomorrow%
    }\\[-1.5ex]
    \multicolumn{4}{l}{%
      ${}^3$~$G = \{\stack(u)\} \cup \{u' \mid \stack(u') = \stack(u) \text{ and $u'$ is a shift node}\}$
    }\\[-1.5ex]
    \multicolumn{4}{l}{%
      ${}^4$~See \cref{def:pending,def:equivalent}.
    }\\
    \bottomrule
  \end{tabularx}
\end{table}

In contrast to rejecting rules, there is only one simple \emph{accepting rule}:
$u$ is accepted when $\Gamma(u)=\set{\#}$ and $\stack(u)=\bot$.

If no termination rules fire on a step node $u$, the construction can proceed by
a \emph{temporal step}. To understand how it works, we need the following
notation: given a node $u$ and a unary temporal operator $\odot$, we denote the
set of all the formulas that appear as arguments of $\odot$ inside $\Gamma(u)$
as $\Gargs_\odot(u) = \set{\alpha \mid \mathop{\odot} \alpha \in \Gamma(u)}$,
and for a set of operators $\set{\odot_1,\ldots,\odot_n}$ we define
$\Gargs_{\odot_1,\ldots,\odot_n}(u)=\Gargs_{\odot_1}(u)\cup\ldots\cup\Gargs_{\odot_n}(u)$.
The temporal step consists in two parts: the application of one \emph{step}
rule, and of one \emph{guess} rule. The step rules, summarised in
\cref{table:steps}, are chosen depending on the type of the leaf at hand, and of
its closest step ancestor. Each rule adds exactly one child $u'$ to the leaf
$u$, whose label is described in the table. The child $u'$ is then fed to one of
the \emph{guess} rules described in \cref{table:guess}. The applicability of the
guess rules depend on the type of $u$ and some other conditions, in a way such
that in each case at most one guess rule is applicable to $u'$. If any is
applicable, the selected rule defines a set of formulas $\Gargs$ as described in
the table, and for each $G\subseteq\Gargs$ adds a child $u''_G$ such that
$\Gamma(u''_G)=\Gamma(u')\cup G$, $\smb(u''_G)=\smb(u')$,
$\stack(u''_G)=\stack(u')$, and $\ctx(u''_G)=\ctx(u')$. After the temporal step
is completed, the construction continues with the expansion rules again, and
everything repeats. 

We can now sketch a soundness and termination argument for the tableau.

\begin{restatable}[Soundness]{theorem}{soundnessthm}
  If the tableau for $\phi$ has an accepted branch, then $\phi$ is satisfiable.
\end{restatable}
\begin{proof}[Sketch]
  $\clos{\phi}$ is finite, and so is the number of possible node labels. Thus,
  unless they are rejected by a rule other than \ref{rule:prune}, all branches
  of the tableau must eventually reach a node that is \emph{equivalent} (cf.\
  \cref{def:equivalent}) to a previous one. Then, they are rejected by
  \cref{rule:prune}. Thus, once fully expanded, the tableau for a formula $\phi$
  is also finite. Then, soundness of the tableau can be proved by building a
  word out of any accepted tableau branch, with a mapping from push and shift
  step nodes of the branch to letters in the word. Chain supports in the word
  correspond to sequences of step nodes.
  See \ifthenelse{\boolean{arxiv}}{\cref{app:soundness}}{\cite{arxiv}}
  for the full proof.
\end{proof}

\begin{table}[t]
  \caption{Step rules}
  \label{table:steps}
  \centering
  \begin{tabular}{c *{5}{@{\hspace{1em}} c}}\toprule
    $u$ & $u_s$${}^1$ & $\Gamma(u')$ & $\smb(u')$ & $\stack(u')$ & $\ctx(u')$ \\\midrule
    push & push/shift &
      $\Gargs_{\ldnext,\lunext}(u)$ &
      $\struct(\Gamma(u))$ & $u$ & $u_s$ or $\bot$${}^2$ \\
    push & pop &
      $\Gargs_{\ldnext,\lunext}(u)$ &
      $\struct(\Gamma(u))$ & $u$ & $\ctx(u_s)$ or $\bot$${}^2$ \\
    shift &  &
      $\Gargs_{\ldnext,\lunext}(u)$ &
      $\struct(\Gamma(u))$ & $\stack(u)$ & $\ctx(u)$ \\
    pop &  &
      $\Gamma(u)$ &
      $\smb(\stack(u))$ & $\stack(\stack(u))$ & $\ctx(\stack(u))$ \\
    \bottomrule\\[-1.75ex]
    \multicolumn{5}{l}{
      ${}^1$~$u_s$ is the closest step ancestor of $u$
    }\\
    \multicolumn{5}{l}{
      ${}^2$~$\ctx(u')=\bot$ if $\stack(u)=\bot$
    }\\
    \bottomrule
  \end{tabular}
\end{table}

\begin{table}[t]
  \caption{Guess rules}
  \label{table:guess}
  \centering
  \begin{tabular}{c @{\hspace{1em}} c @{\hspace{1em}} c}\toprule
    $u$ & only if & $\Gargs$ \\\midrule
    push/shift & & $\mathcal{G}_{\lwdnext}(u_s)\cup\mathcal{G}_{\lwunext}(u_s)
    \cup\mathcal{G}_{\lhdnext}(u_s)\cup\mathcal{G}_{\lwhdnext}(u_s)$ \\
    pop & $u_c\ne\bot$${}^1$ & 
    $\bigcup\begin{cases}
      \mathcal{G}_{\lcdnext}(u_c) \cup \mathcal{G}_{\lcunext}(u_c)  \\
      \mathcal{G}_{\lwcdnext}(u_c) \cup \mathcal{G}_{\lwcunext}(u_c) \\
      \mathcal{G}_{\lhdnext}(u_c) \cup \mathcal{G}_{\lwhdnext}(u_c) \\
      \mathcal{G}_{\lwhunext}(\stack(u_s)) \cup \mathcal{G}_{\lhunext}(\stack(u_s))\end{cases}$\\\midrule
    \multicolumn{3}{l}{${}^1$~$u_s$ is the closest step ancestor of $u$ and $u_c=\ctx(u_s)$}\\\bottomrule
  \end{tabular}
\end{table}


%% file: sections/4.encoding.tex

\section{SMT Encoding of the Tableau}
\label{sec:encoding}

Our technique for symbolic model checking of \ac{POTL_f} properties does not directly
construct the tableau described in \cref{sec:tableau}, but rather, it
\emph{encodes} it into SMT formulas that can be efficiently handled by
off-the-shelf solvers. Iterating over a growing index $k > 1$, at each step our
procedure produces an SMT formula that encodes the branches of the tableau of
length up to $k$ step nodes, such that the formula is satisfiable if and only if
an accepted branch of the tableau exists. If not, we increment $k$ and proceed.
In this respect, the procedure reminds of classic \emph{bounded model
checking}~\cite{DBLP:journals/ac/BiereCCSZ03,ClarkeBRZ01}.
Here we summarize the working principles of the tableau encoding. The full
details are available in \ifthenelse{\boolean{arxiv}}{\cref{app:encoding}}{\cite{arxiv}}.

The encoding produces formulas whose \emph{models}, when they exist, represent
single branches of the tableau. At a given step $k$, the formulas are
interpreted over a restricted form of quantified\footnote{Thanks to finite
sorts, quantifiers are in fact expanded to disjunctions/conjunctions.} EUF, over
two \emph{finite, enumerated, ordered}\footnotemark\ sorts: a sort $\N_k$, of
exactly $k+1$ elements used to identify the nodes in the branch, and a sort
called $\S$ that contains a finite set of symbols used in the encoding to
represent the letters of the formula's alphabet. We suppose to have a finite
number of constants for the values in $\S$. Among those, we have $p\in\S$ for
each $p\in\Sigma\cup AP$. Others will be introduced when needed. We also exploit
a fixed arbitrary ordering between elements of $\N_k$, and we abuse notation by
denoting the constants for sort $\N_k$ as $0,1,\ldots,k$, and writing $x+1$ and
$x-1$ for an element $x\in\N_k$ to denote its predecessor and successor in this
order.

\footnotetext{%
    The sort returned by the \texttt{Z3\_mk\_finite\_domain\_sort()} function of
    the Z3 C API.%
}

For each proposition $p\in\Sigma\cup AP$, the encoding uses a binary predicate
$\Gamma(p, x)$ whose first argument ranges among $\S$ and the second among
$\N_k$. The intuitive meaning of $\Gamma(p,x)$ is that $p\in\Gamma(u)$ if $u$ is
the $x$-th step node of the current branch of the tableau. The encoding also
uses some function symbols. A unary predicate $\bar\Sigma$ ranging over $\S$
tells which symbols from $\S$ are structural symbols. A function
$\smb(x):\N_k\to \S$ is used to represent the $\smb(u_x)$ symbol. A function
symbol $\struct(x):\N_k\to\S$ represents $\Gamma(u_x) \cap \Sigma$. Two
functions $\stack(x) : \N_k \to \N_k$ and $\ctx(x) : \N_k \to \N_k$ represent
the corresponding functions in the tableau. When $\stack(u)=\bot$, we denote it
as $\stack(x)=0$, and similarly for $\ctx(x)$.

For any strong or weak \emph{next} or \emph{chain next} temporal formula in
the closure of $\phi$ we also introduce a corresponding
\emph{propositional} symbol in $\S$.  Specifically, for each formula
$\lnextsup{t}\psi$, $\lcnext{t}\psi$, $\lwnextsup{t}\psi$ and
$\lwcnext{t}\psi$ in the closure, $\S$ contains the following propositional
symbols, which we call \emph{grounded}: $(\lnextsup{t}\psi)_G$,
$(\lcnext{t}\psi)_G$, $(\lwnextsup{t}\psi)_G$, $(\lwcnext{t}\psi)_G$, and
$(\lnextsup{t}\psi)_G$, $(\lcnext{t}\psi)_G$, $(\lwnextsup{t}\psi)_G$,
$(\lwcnext{t}\psi)_G$.

The core building block of the encoding is the following \emph{normal form} for
\ac{POTL_f} formulas.
\begin{definition}[Next Normal Form]
  \label{def:encoding:xnf}
  Let $\phi$ be a \ac{POTL_f} formula. The \emph{next normal form} of $\phi$, denoted 
  $\xnf(\phi)$ is defined as follows:
  \begin{align*}
    &
    \begin{aligned}
      &\xnf(p) = p\quad\text{for $p\in\Sigma$} &
      &\qquad\qquad\xnf(\neg p) = \neg p\quad\text{for $p\in\Sigma$}\\
      &\xnf(\lwnextsup{t}\psi) = \lwnextsup{t}\psi &
      &\qquad\qquad\xnf(\lwcnext{t}\psi) = \lwcnext{t}\psi\\
    \end{aligned}\\
    &\xnf(\alpha\circ\beta) = 
      \xnf(\alpha)\circ\xnf(\beta)\quad\text{for $\circ\in\{\lor,\land\}$}\\
    &\xnf(\lcuntil{t}{\alpha}{\beta}) = 
      \xnf(\beta) \lor \big( 
        \xnf(\alpha) \land (
          \lnextsup{t}(\lcuntil{t}{\alpha}{\beta}) \lor
          \lcnext{t}(\lcuntil{t}{\alpha}{\beta})
        )
      \big)\\
    &\xnf(\lcrelease{t}{\alpha}{\beta}) = 
      \xnf(\beta) \land \big( 
        \xnf(\alpha) \lor (
          \lwnextsup{t}(\lcrelease{t}{\alpha}{\beta}) \land
          \lwcnext{t}(\lcrelease{t}{\alpha}{\beta})
        )
      \big)
  \end{align*}
\end{definition}
Intuitively, $\xnf(\phi)$ encodes the \emph{expansion rules} of the tableau
(\cref{table:expansion}). Given $\phi$ and a fresh variable $x$ of sort $\N_k$,
we denote as $\xnf(\phi)_G$ the formula obtained from $\xnf(\phi)$ by replacing
any proposition $p$ with $\Gamma(p, x)$. Note that $\xnf(\phi)_G$ does not
contain temporal operators: it is a first-order formula with a single free
variable $x$.

We can now show the encoding itself. We start by constraining
the meaning of the $\bar\Sigma$ predicate and the $\struct$ and $\smb$
functions. We define a formula $\phi_{\mathrm{axioms}}$ that states that the
$\bar\Sigma$ predicate identifies structural symbols and the $\struct(x)$ and
$\smb(x)$ functions only return structural symbols, and we write a formula
$\phi_{\mathrm{OPM}}$ that explicitly models the $\lessdot$, $\doteq$ and
$\gtrdot$ relations between symbols in $\S$ as binary predicates in the SMT
encoding. The predicates range over the whole $\S$ but only the relationship
between symbols in $\Sigma$ will matter. With these in place, we can identify
the \emph{type} of each step node depending on the PR between $\smb(x)$
and $\struct(x)$. We encode this by the following three predicates:
\begin{gather*}
  \push(x) \equiv \smb(x) \lessdot \struct(x)\quad
  \shift(x) \equiv \smb(x) \doteq \struct(x)\\
  \pop(x) \equiv \smb(x) \gtrdot \struct(x)
\end{gather*}

A formula $\phi_{\mathrm{init}}$ encodes how the root node of the tableau looks
like. In particular, it includes the conjunct $\xnf(\phi)_G(1)$, to say that
its label contains $\phi$.

We can now encode the step rules of \cref{table:steps}. For space constraints we
only show here the encoding of the step rules concerning \emph{push} nodes
(first two lines of \cref{table:steps}). The encoding of such rules is the
following:
\begin{align*}
  \steprule_{\push}(x) \equiv {} & 
  \smashoperator{\bigwedge_{\lnextsup{t}\alpha\in\clos{\phi}}}\big(
    \Gamma((\lnextsup{t}\alpha)_G, x) \implies \xnf(\alpha)_G(x+1)
  \big)\\
  {} \land {} & \smb(x+1)=\struct(x) \land \stack(x+1)=x \\
  {} \land {} & (\stack(x)=0 \implies \ctx(x+1)=0) \\
  {} \land {} & ((\stack(x)\ne 0 \land (\push(x-1) \lor \shift(x-1))) \implies \ctx(x+1)=x-1) \\
  {} \land {} & ((\stack(x)\ne 0 \land \pop(x-1)) 
    \implies \ctx(x+1)=\ctx(x-1))
\end{align*}

We can similarly obtain two formulas $\steprule_{\shift}(x)$ and
$\steprule_{\pop}(x)$. It is worth to note the first line of the above
definition, where $\xnf(\alpha)$ is imposed to hold on $x+1$ if a next operator
on $\alpha$ is present on $x$.

Next, we can encode the rejecting rules of \cref{table:rejecting}. Since there
are so many of them, we only show some examples
(see \ifthenelse{\boolean{arxiv}}{\cref{app:encoding}}{\cite{arxiv}} for the full list).
What we actually encode is the \emph{negation} of the rejecting rules,
that describes what a node has to satisfy to \emph{not} be rejected. We start to
note that \cref{rule:contradiction} does not need to be encoded, since it just
states that a proposition cannot hold together with its negation, which is
trivially implied by the logic. Then, the simplest ones are
\cref{rule:end,rule:conflict} of \cref{table:rejecting}, and can be encoded as
follows:
\begin{align*}
  \mathrm{r_{\ref{rule:conflict}}}(x) \equiv {} & \forall p\, \forall q (\Sigma(p) \land \Sigma(q) 
  \land \Gamma(p, x) \land \Gamma(q, x) \implies p = q) \\
  \mathrm{r_{\ref{rule:end}}}(x) \equiv {} & \Gamma(\#, x) \implies \bigl(
    \smashoperator{\bigwedge_{\lnextsup{t}\alpha\in\clos{\phi}}} 
      (\neg \Gamma((\lnextsup{t}\alpha)_G, x)) \land
      \bigwedge_{p\in AP} (\neg \Gamma(p, x))
  \bigr)
\end{align*}

We similarly have a formula $r_i(x)$ encoding the negation of each block of
lines from \cref{rule:check:pnext} to \ref{rule:prune}. With these in place, we
define a formula $\unr{\phi}_k$ called the $k$-\emph{unraveling} of $\phi$, that
encodes all the non-rejected branches of the tableau of up to $k$ step nodes.
\begin{align*}
  &\phi_{\mathrm{axioms}} \land \phi_{\mathit{OPM}} \land \phi_{\mathrm{init}}
  \land 
  \forall x \left(x > 1\implies \bigwedge\nolimits_{i=2}^{13} r_i(x)\right)  \land {}\\
  &\forall x \left[
    1\le x < k 
     \implies \left(
    \begin{aligned}
      & (\push(x) \implies \steprule_{\push}(x)) \\
      {}\land{}& (\shift(x) \implies \steprule_{\shift}(x)) \\
      {}\land{}& (\pop(x) \implies \steprule_{\pop}(x))
    \end{aligned}\right)
  \right]
\end{align*}
The only \emph{acceptance} rule of the tableau is encoded by a
formula $e(x)$ defined as $e(x)\equiv \Gamma(\#,x)\land\stack(x)=0$.

Finally, we have the following.
\begin{theorem}
  If $\unr{\phi}_k\land e(k)$ is satisfiable for some $k>0$, then $\phi$ is
  satisfiable.
\end{theorem}

We exploit this encoding of \ac{POTL_f} satisfiability for model checking a formula $\phi$
through an algorithm that iterates on $k$ starting from $k=1$.
First, we check satisfiability of $\unr{\neg \phi}_k \land \unr{\mathcal{M}}_k$,
where $\unr{\mathcal{M}}_k$ encodes a length-$k$ prefix of a trace
of the program $\mathcal{M}$ to be checked.
We automatically translate programs to OPA whose transitions are labeled with program statements
in the same way as~\cite{AlurBE18,ChiariMPP23},
so that the automaton's stack simulates the program stack.
Such \emph{extended} OPA are then directly encoded into SMT in a straightforward manner,
using the theories of fixed-size bit vectors and arrays to represent variables
(cf.\ \ifthenelse{\boolean{arxiv}}{\cref{app:program-encoding}}{\cite{arxiv}}).
If this satisfiability check fails, it means no trace of $\mathcal{M}$
of length $\geq k$ violates $\phi$, proving that $\mathcal{M}$ satisfies $\phi$.
Otherwise, we check whether $e(k)$ is satisfied when conjoined with the previous assertions.
If it is, then we have found a counterexample trace that violates $\phi$.
Otherwise, we increase $k$ by 1 and repeat.
Since the tableau is finite, we eventually either find a counterexample,
or hit a value of $k$ such that \cref{rule:prune} rejects all branches,
and the initial satisfiability check fails.


%% file: sections/5.evaluation.tex
\section{Experimental Evaluation}
\label{sec:exp:eval}

We implemented the encoding described in~\cref{sec:encoding} in a SMT-based
model checker that leverages the Z3 SMT
solver~\cite{DBLP:conf/tacas/MouraB08}.
We developed it within POMC~\cite{pomc},
an explicit-state model checker for \ac{POTL} developed by the authors of \cite{ChiariMPP23}.

We compare our SMT-based approach with the explicit-state algorithm powering POMC,
which performs the following steps on-the-fly:
\begin{enumerate*}[label=(\roman*)]
  \item it builds an OPA $\mathcal{A}_\varphi$ encoding the negation
    of the formula $\varphi$ to be checked;
  \item it constructs the synchronized product between $\mathcal{A}_\varphi$
    and the model of the system;
  \item it checks the nonemptiness of the product automaton, witnessing
    a counterexample to the property in the model, in a depth-first
    fashion.
\end{enumerate*}

We ran our experiments on server with a 2.0 GHz AMD CPU and RAM capped at 30 GiB.

\subsection{Description of the benchmarks}
\label{sub:exp:eval:bench}
We evaluate the two tools on a set of benchmarks adapted from
\cite{ChiariMPP23}, divided in three categories (Quicksort, Jensen, Stack).
We modeled all benchmarks in MiniProc, the modeling language of
the POMC tool. The checked formulas are reported in
Table~\ref{table:bench}. Below, we give a brief description of each
category.

\begin{table}[tb]
  \small
  \centering
  \caption{Benchmark formulas.
    The last column states whether they are true (T) or false (F) in each
    model. 
    $\llglob$ is the LTL always, which we implemented as in~\cite{GeattiGM19}.}
  \label{table:bench}
  \renewcommand\arraystretch{1.2}%
  \begin{tabularx}{\textwidth}{c c @{\hspace{.5em}} X r}\toprule
  \parbox[t]{1.2em}{\multirow{6}{*}{\rotatebox[origin=c]{90}{QuickSort}}}
  & 1  & $\lcunext (\lret \land \mathrm{main})$ & T \\
  & 2  & $\lcall \land \mathrm{main} \implies \neg (\lunext \lexc \lor \lcunext \lexc)$ & T \\
  & 3  & $\llglob (\lcall \land \mathrm{qs} \implies \neg (\lunext \lexc \lor \lcunext \lexc))$ & F \\
  & 4  & $\lcunext \mathrm{sorted}$ & F \\
  & 5  & $\llglob (\lcall \land \mathrm{qs} \implies \lcunext \mathrm{sorted})$ & F \\
  & 6  & $\lcdnext (\lhan \land \ldnext (\lcall \land \mathrm{qs} \land \lcunext (\lexc \lor \mathrm{sorted})))$ & T \\
  \midrule
  \parbox[t]{1.2em}{\multirow{4}{*}{\rotatebox[origin=c]{90}{Jensen}}}
  & 8  & $\llglob (\lcall \land \neg P_\mathrm{cp} \implies \neg (\lcduntil{\top}{(\lcall \land \mathrm{read}))})$ & T \\
  & 9  & $\llglob (\lcall \land \neg P_\mathrm{db} \implies \neg (\lcduntil{\top}{(\lcall \land \mathrm{write}))})$ & T \\
  & 10 & $\llglob (\lcall \land ((\mathrm{canpay} \land \neg P_\mathrm{cp}) \lor (\mathrm{debit} \land \neg P_\mathrm{db})) \implies \lunext \lexc \lor \lcunext \lexc)$ & T \\
  & 11 & $\neg (\lcduntil{\top}{(\mathrm{balance} < 0)})$ & T \\
  \midrule
  \parbox[t]{1.2em}{\multirow{4}{*}{\rotatebox[origin=c]{90}{Stack}}}
  & 12 & $\llglob (\mathrm{modified} \implies \neg (\lunext \lexc \lor \lcunext \lexc))$ & T/F \\
  & 13 & $\llglob (\lcall \land (\mathrm{push} \lor \mathrm{pop}) \implies \neg (\lhduntil{\top}{\mathrm{modified}}))$ & T/F \\
  & 14 & $\llglob (\lcall \land (\mathrm{push} \lor \mathrm{pop}) \land \lcdnext \lret \implies$ \newline
         \hspace*{2.5cm} $\neg (\lcduntil{\top}{(\lhan \land \mathrm{Stack} \land (\lcduntil{\neg \lhan}{(\mathrm{T} \land \lunext \lexc)}))}))$ & T/T \\
  \bottomrule
  \end{tabularx}
\end{table}

\paragraph{Quicksort.}
We modeled a Java implementation of the Quicksort sorting algorithm.
The algorithm is implemented as a recursive function $\mathrm{qs}$,
called by the $\mathrm{main}$ function in a \texttt{try-catch} block,
and is applied to an array of integers that may contain
null values, which cause a \texttt{NullPointerException}.
We vary the length of the arrays from 1 to 5 elements
and the width of the elements from 2 to 16 bits.
Formulas 1 and 2 both check that the $\mathrm{main}$ function returns without exceptions,
while 3 checks the same for the $\mathrm{qs}$ (QuickSort) function.
Formulas 4 (resp., Formula 5) states that the array is sorted when the main
function (resp., the $\mathrm{qs}$ function) returns without exceptions.
Finally, Formula 6 states that either $\mathrm{qs}$ throws an exception or
the array is sorted (and $\mathrm{qs}$ returns normally).

\paragraph{Bank Account.}
This category consists of a simple banking application taken from
\cite{JensenLT99} which allows users to withdraw money or check their
balance.  The variable representing the balance is protected by a Java
AccessController, which prevents unauthorized users from accessing it by
raising exceptions. We modeled the balance with an integer variable.
Formula 8 (resp., Formula 9) checks that, whenever a function is called
without having permission to check the balance (resp., to make a payment),
then there is no read-access (resp., write access) to the variable holding
the balance.  The permission of checking the balance and to make a payment
are modeled by the variables $P_\mathrm{cp}$ and $P_\mathrm{db}$,
respectively.
Formula 10 checks that if the functions that check the balance
($\mathrm{canpay}$) and make a payment ($\mathrm{debit}$) are called
without permission, an exception is thrown. Formula 11 checks that the
balance never becomes negative, because payments are only made if the
account has enough money.

\paragraph{Stack.}
We model two C++ implementations of a generic stack data structure taken
from \cite{Sutter97}, where constructors of contained elements may throw
exceptions. Only one of the two implementations is exception safe.  The
\texttt{pop} method of the safe implementation does not return the popped
element, which must be accessed through the \texttt{top} method, and it
performs other operations on a new copy of the internal data structure, to
prevent exceptions from leaving it in an inconsistent state.
In contrast with~\cite{ChiariMPP23} which uses a manually-crafted
abstraction for the elements in the stack, our model implements the stack
with actual arrays of fixed-width integers.
Formulas 12 and 13 check \emph{strong exception safety}~\cite{Abrahams00},
\ie that each operation on the data structure is rolled back if any
functions related to the element type $\mathrm{T}$ throw an exception,
leaving the stack in a consistent state.
Formula 14 checks \emph{exception neutrality}~\cite{Abrahams00}, which
means that exceptions thrown by element functions are always propagated by
the stack's methods.

\setlength{\abovecaptionskip}{0pt}
\setlength{\belowcaptionskip}{0pt}

\begin{figure}[p]
\centering
\begin{subfigure}{1\textwidth}
    \centering
    \includegraphics[width=.8\textwidth, height=6cm]{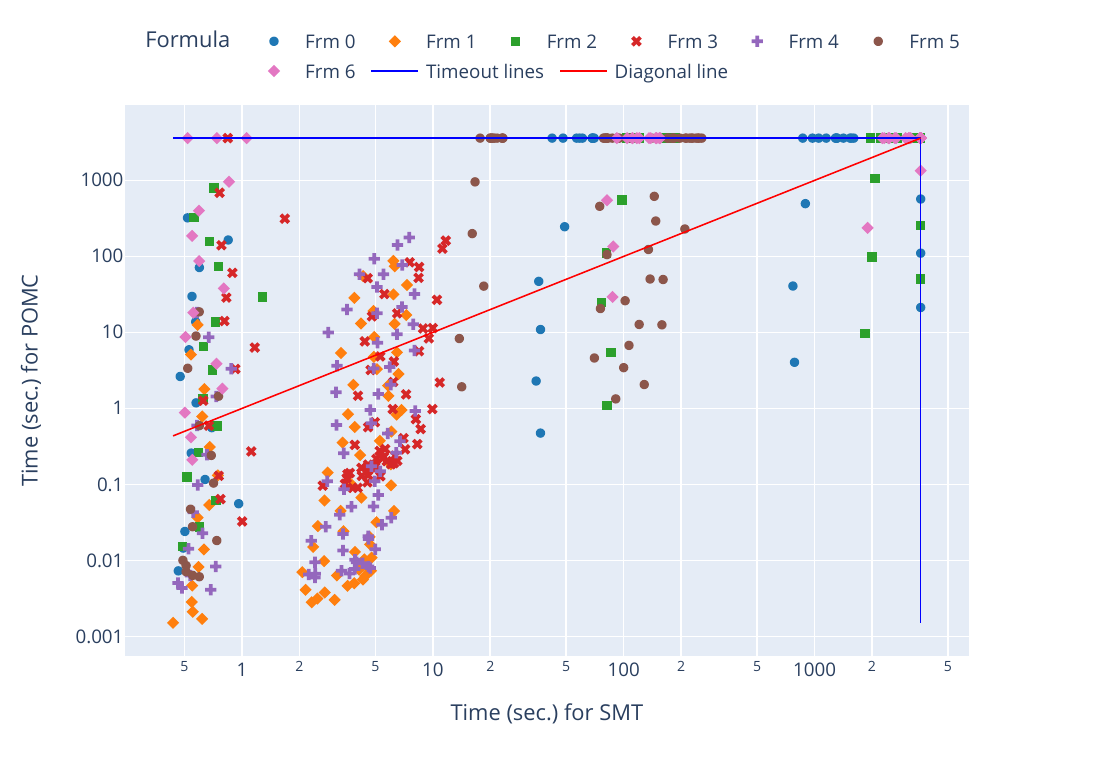}
    \caption{Scatter plot for category Quicksort.}
\end{subfigure}
\begin{subfigure}{1\textwidth}
    \centering
    \includegraphics[width=.8\textwidth, height=6cm]{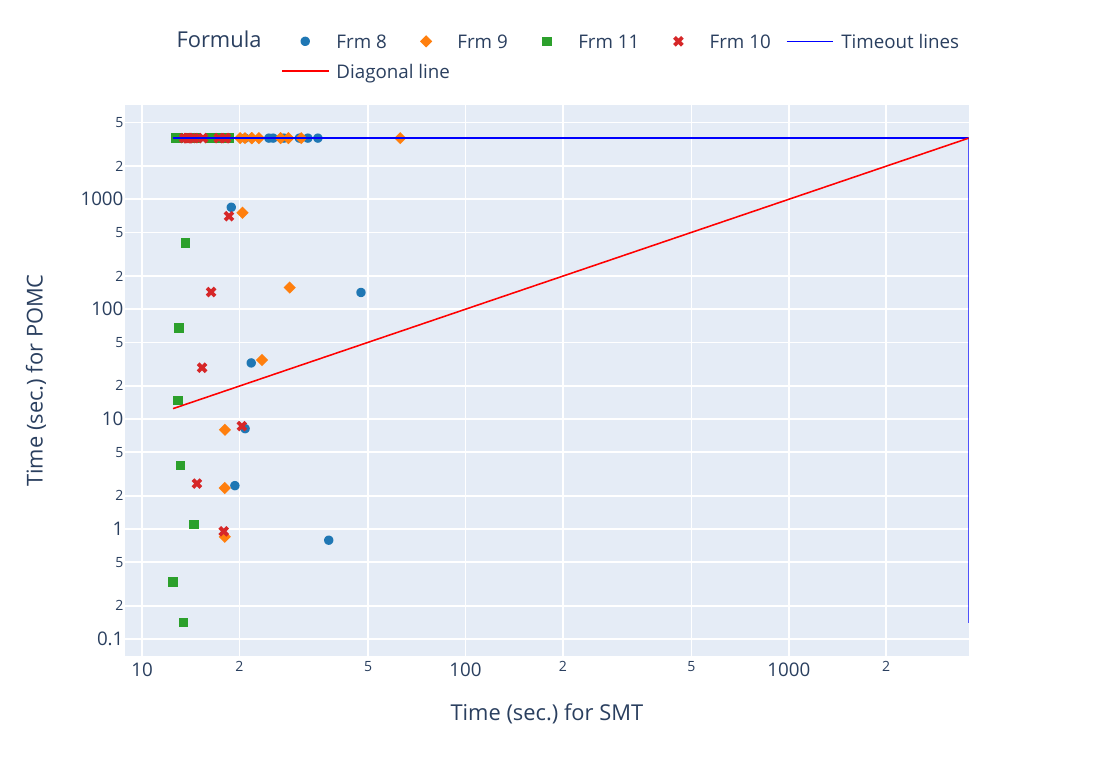}
    \caption{Scatter plot for category Jensen.}
\end{subfigure}
\begin{subfigure}{1\textwidth}
    \centering
    \includegraphics[width=.8\textwidth, height=6cm]{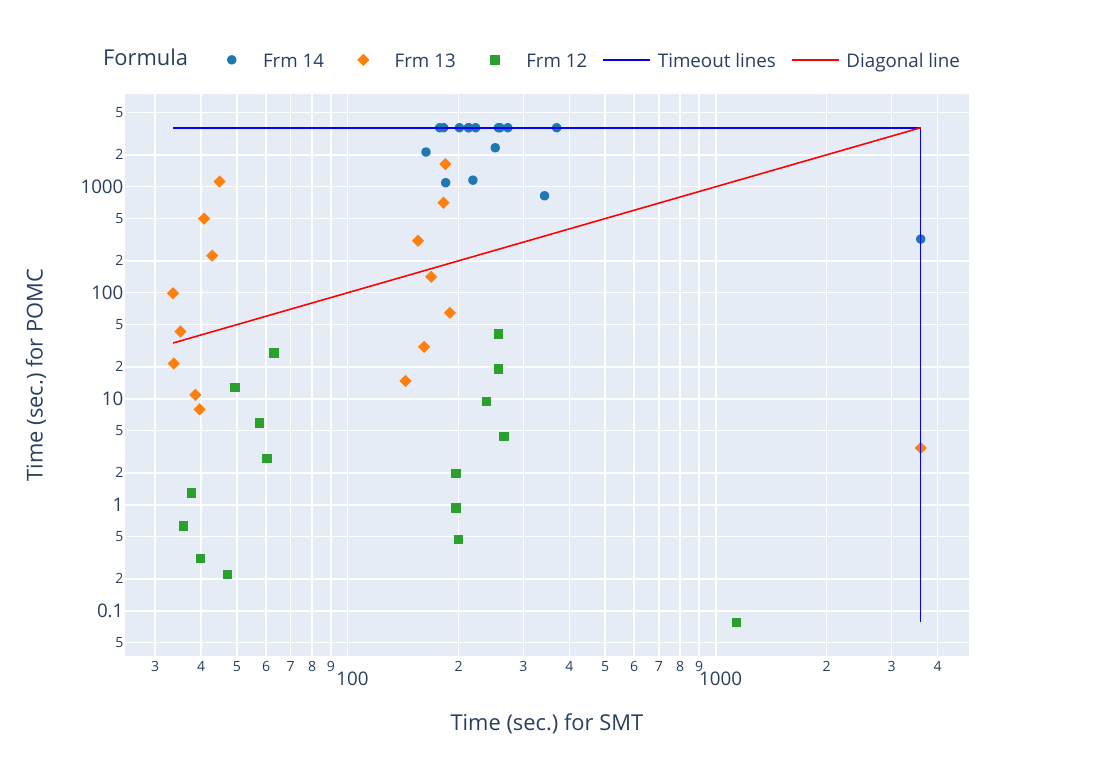}
    \caption{Scatter plot for category Stack.}
\end{subfigure}
\caption{Scatter plots}
\label{fig:scatter:plots}
\end{figure}

\begin{figure}[pt]
\centering
\includegraphics[width=1\linewidth]{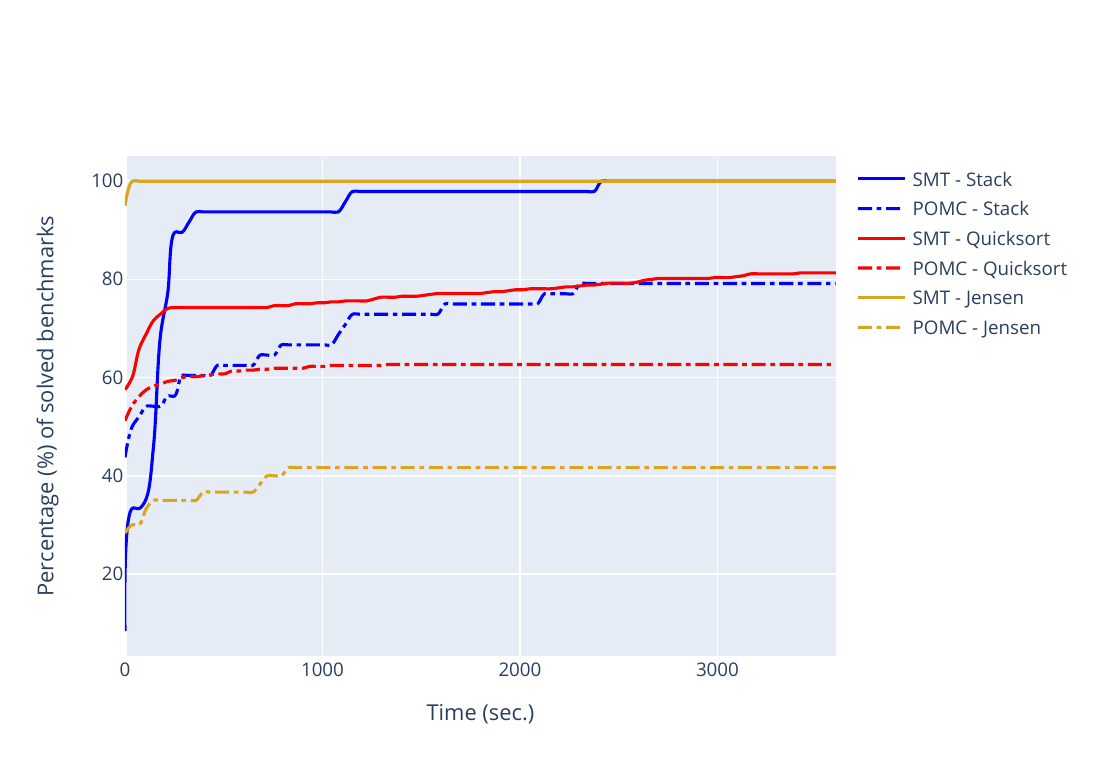}
\caption{Survival plot}
\label{fig:survival:plots}
\end{figure}

\subsection{Description of the plots}

We compare the time (measured in seconds) taken by the SMT-based approach
(in the plots referred to as SMT) with the time taken by POMC,
dividing the plots by the three categories of benchmarks (Quicksort,
Jensen, and Stack). For each category, we show a scatter plot
(\cref{fig:scatter:plots}) and a survival plot (\cref{fig:survival:plots}).

We first look at the scatter plots in~\cref{fig:scatter:plots}. The x-axis
refers to the solving time for the SMT-based approach while the y-axis to
the solving time for POMC, both measured in seconds.
The blue border lines indicate the timeout (set to 3600 seconds) for the
tools, while the red line denotes the diagonal of the plot.

For all three categories of benchmarks, the scatter plots reveal an exponential
blow up for the solving time of the POMC tool; on the contrary, the
SMT-based approach does not incur in such a blow up. As an example, we take
the scatter plot for the Quicksort category in~\cref{fig:scatter:plots}
(a) and we consider the brown circles in the middle of the plot,
corresponding to the Formula 5 of Table~\ref{table:bench} checked on an
array of size 2 containing numbers of increasing bitvector-size.  For the
case of numbers of bitvector-size of 3, 4, 5, and 6 bits, the solving time
of POMC is of 8, 40, 199, and 956 seconds, respectively, while the time
required by the SMT-based approach is of 13, 18, 16 and 16 seconds,
respectively. Moreover, while for bitvector-size greater than 6 bits POMC
reaches always the timeout for Formula 5, the SMT-based approach solves the
benchmarks of all bitvector-size (\ie up to 16 bits) in time always less
than 23 seconds.

A similar consideration can be done for the Jensen and the Stack
categories. Take, for example, the blue circles in~\cref{fig:scatter:plots}
(c) corresponding to Formula 14 in Table~\ref{table:bench}. For this case,
the solving times of the SMT-based approach are consistently better than
the ones of POMC. 
  The reason may be that this formula contains hierarchical operators,
  which tend to yield to automata that make more non-deterministic guesses.
This, in turn, causes the explicit-state model checker to perform, in
general, many steps of backtracking during its depth-first model checking
algorithm. Conversely, in the SMT-based approach, this part is managed
(efficiently) by the DPLL algorithm inside the SMT-solver.

\begin{figure}[tb]
\centering
\includegraphics[width=1\linewidth]{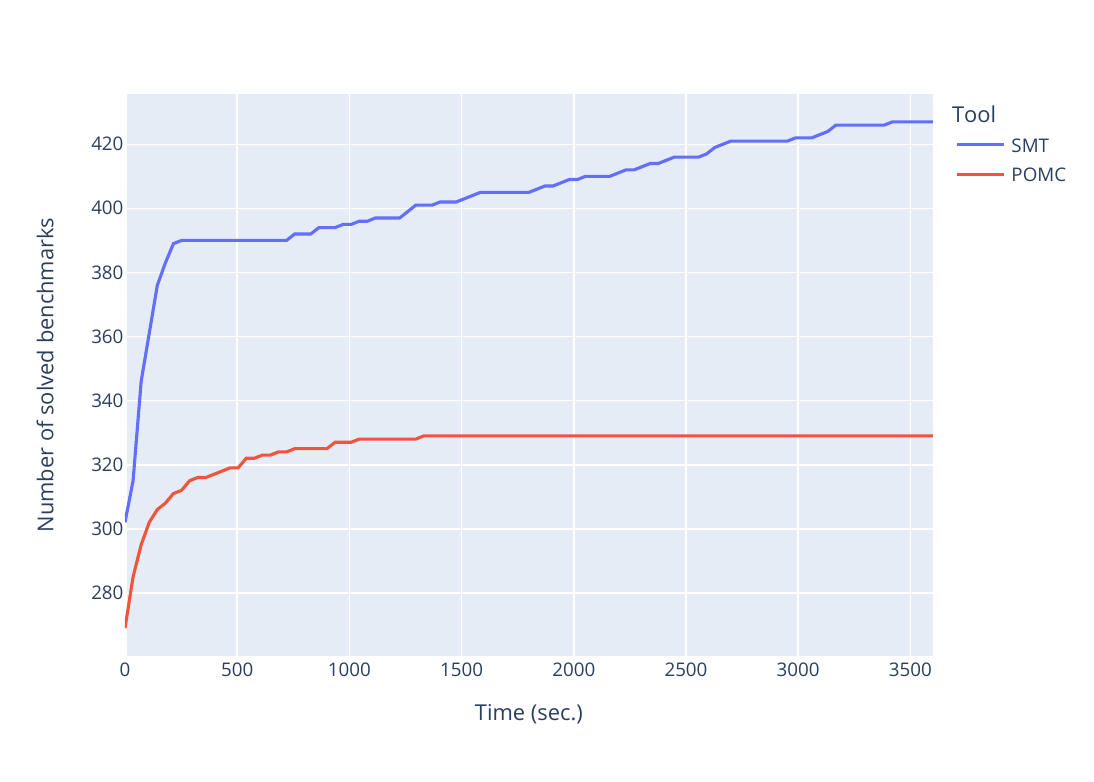}
\caption{Survival plot for the Quicksort category.}
\label{fig:survival:plots:quicksort}
\end{figure}

The exponential trend of POMC is reflected also in the survival plot
(\cref{fig:survival:plots}). Here, the x-axis represents the time (in
seconds) while the y-axis represents the percentage of solved benchmarks.
From the blue and yellow lines in~\cref{fig:survival:plots}, which
correspond to the categories Stack and Jensen, respectively, it is clear
that the POMC tool gets stuck solving (approximately) the 80\% and the 60\%
of the benchmarks in the corresponding category. Conversely, the SMT-based
approach solves all benchmarks in these two categories. If we take a look
to the survival plot only for the Quicksort category in \cref{fig:survival:plots:quicksort}
(which reports the absolute number of solved benchmarks), we observe that
the POMC tool gets stuck solving (approximately) 330 benchmarks, while the
SMT-based approach solves circa
430 benchmarks.

In our benchmarks, we found only one case in which the solving time of POMC
is always better than the one of the SMT-based approach. It corresponds to
the green squares on the scatter plots in~\cref{fig:scatter:plots} (c) for
the Stack category, corresponding to Formula 12. The reason is that this
formula requires very few nondeterministic transitions in the
explicit-state automaton. This, in turn, makes the search of the
state-space a (almost) deterministic step, and thus very efficient for the
depth-first algorithm of POMC.  On the contrary, the breadth-first
algorithm of the SMT-based approach seems to perform worse.


%% file: sections/6.conclusions.tex

\section{Conclusions}
\label{sec:conclusions}

We have introduced a tree-shaped tableau for the future fragment of the temporal logic \ac{POTL} on finite-word semantics,
and encoded it in SMT to perform symbolic model checking of procedural programs.
This is the first time both of these techniques have been used for checking a temporal logic with context-free modalities.
The experimental evaluation shows that our symbolic approach scales better than the state-of-the-art explicit-state one.

Extending the tableau to past \ac{POTL} operators and to infinite words
seems a promising future direction, which should be achievable through an
approach similar to related work on the tree-shaped tableau for
LTL~\cite{GeattiGMR21}.


%% file: sections/Appendix.tex
\section{Appendix: Soundness of the Tableau}
\label{app:soundness}

The correctness proof is based on the fact that an accepting branch of the tableau
identifies an OP word which is a model for formula $\phi$.

In the following, we let $[n:m] = \set{n, n+1, \dots, m}$.
Let $\bar{u} = u_0 u_1 \dots u_m$ be an accepting branch,
and $\bar{\pi} = \pi_0 \pi_1 \dots \pi_n$
the sequence of all step nodes in $\bar{u}$.
For any node $u$ we define its propositional label as $\Gamma_{AP}(u) = \Gamma(u) \cap (\Sigma \cup AP)$.
We define an OP word $w_{\bar{u}} = \# \Gamma_{AP}(\pi_{i_0}) \Gamma_{AP}(\pi_{i_1}) \dots \allowbreak \Gamma_{AP}(\pi_{i_k})$
such that $i_0 = 0$, $i_k = n$ and for all $p \in [1:k-1]$,
$i_p$ is the index of the first push or shift node descending from $\pi_{i_{p-1}}$.
We define a function $\iota$ that relates positions in $w$ to step nodes,
such that $\iota(0) = \bot$ and $\iota(p) = \pi_{i_p}$ for all $p \in [1:k]$.

The OP alphabet $(\Sigma, M)$ of the tableau defines the chain relation for word $w_{\bar{u}}$.
To see how chains in $w_{\bar{u}}$ are reflected in $\bar{\pi}$,
we define the concept of \emph{support} for tableau branches.

We call a \emph{support} for the simple chain
$\ochain {c_0} {c_1 c_2 \dots c_\ell} {c_{\ell+1}}$
a sequence of step nodes of the form
$\pi_1 \pi_2 \dots \pi_{\ell-1} \pi_{\ell} \pi_{\ell+1}$
such that
\begin{itemize}
\item for all $p \in [1:\ell+1]$, we have $\Gamma_{AP}(\pi_p) = c_p$;
\item $\pi_1$ is a push node, all nodes from $\pi_2$ to $\pi_{\ell}$ are shift nodes,
and $\pi_{\ell+1}$ is a pop node;
\item $\stack(\pi_p) = \pi_1$ for all $p \in [2:\ell+1]$.
\end{itemize}

We call a \emph{support for the composed chain}
$\ochain {c_0} {s_0 c_1 s_1 c_2 \dots c_\ell s_\ell} {c_{\ell+1}}$
a sequence of step nodes of the form
$\bar{\pi}_0 \pi_1 \bar{\pi}_1 \pi_2 \dots \pi_\ell \bar{\pi}_\ell \pi_{\ell+1}$
such that
\begin{itemize}
\item for all $p \in [1:\ell+1]$, we have $\Gamma_{AP}(\pi_{j_p}) = c_p$;
\item $\pi_1$ is a push node, all nodes from $\pi_2$ and $\pi_\ell$ are shift nodes,
  and $\pi_{\ell+1}$ is a pop node;
\item $\stack(\pi_p) = \pi_1$ for all $p \in [2:\ell+1]$;
\item for every $p = [0:\ell]$
  if $s_p \neq \epsilon$, then $\bar{\pi}_p$
  is a support for the chain $\ochain {c_p} {s_p} {c_{p+1}}$,
  else $\bar{\pi}_p$ is the empty sequence.
\end{itemize}

We define the \emph{depth} of a chain as follows:
$\depth(\ochain{a}{x}{b}) = 0$ if $\ochain{a}{x}{b}$ is a simple chain, and
$\depth(\ochain {c_0} {s_0 c_1 \dots c_\ell s_\ell} {c_{\ell+1}}) = 1 + \max_{i = 0, 1, \dots, \ell} \depth(\ochain{c_i}{s_i}{c_{i+1}})$.

\begin{lemma}
  \label{lemma:soundness-help}
Let $\ochain {c_0} {s_0 c_1 s_1 c_2 \dots c_\ell s_\ell} {c_{\ell+1}}$ be a chain,
$\bar{\pi}_0 \pi_1 \bar{\pi}_1 \pi_2 \dots \pi_\ell \bar{\pi}_\ell \pi_{\ell+1}$
its support in $\bar{u}$,
and $\pi_0$ the node in $\bar{\pi}$ corresponding to the word position of $c_0$
(i.e., a push or shift node such that $\Gamma_{AP}(\pi_0) = c_0$).

Then, (a) for all $p \in [2:\ell+1]$ we have $\ctx(\pi_p) = \pi_0$,
and (b) $\ctx(\stack(\pi_{\ell+1})) = \ctx(\pi'_0)$,
where $\pi'_0$ is the successor of $\pi_0$ in $\bar{\pi}$.
\end{lemma}
\begin{proof}
The $\ctx$ of tableau nodes is determined by the step rules in \cref{table:steps}.

Since $\pi_1$ is a push node, the rule to apply depends on the type of its closest step ancestor $u_s$,
which depends in turn on whether $s_0 = \epsilon$.
We prove the claim by induction on
$\depth(\ochain {c_0} {s_0 c_1 s_1 c_2 \dots c_\ell s_\ell} {c_{\ell+1}})$.

The base case is when the chain is simple, so $s_i = \epsilon$ for all $i \in [1:\ell]$.
In this case, when applying the step rules to $\pi_1$,
we have $u_s = \pi_0$, which is either a push or a shift node,
and $\ctx(u') = \ctx(\pi_2) = \pi_0$.
Nodes from $\pi_2$ to $\pi_\ell$ are shift nodes:
the third row of \cref{table:steps} applies and,
by a simple induction, $\ctx(u') = \ctx(\pi_{p+1}) = \ctx(\pi_p) = \pi_0$
for all $p \in [1:\ell]$, which proves (a).
Claim (b) follows from the fact that $\stack(\pi_p) = \pi_1 = \pi'_0$ for all $p \in [2:\ell+1]$
due to the first and third rows of \cref{table:steps},
and in particular $\ctx(\stack(\pi_{\ell+1})) = \ctx(\pi'_0)$.

For the inductive step, we have the hypothesis that claims (a) and (b) are true for all sub-chains
$\ochain {c_p} {s_p} {c_{p+1}}$ with $p \in [0:\ell]$ such that $s_p \neq \epsilon$.

Let $\bar{\pi}^0_0 \pi^0_1 \bar{\pi}^0_1 \pi^0_2 \dots \pi^0_{\ell^0} \bar{\pi}^0_{\ell^0} \pi^0_{\ell^0+1}$,
be the support of chain $\ochain{c_0}{s_0}{c_1}$ in $\bar{\pi}$.
By claim (a), we have $\ctx(\pi^0_{\ell^0+1}) = \pi_0$.
Since $\pi_1$ is a push node and $u_s = \pi^0_{\ell^0+1}$ is a pop node,
the second row of \cref{table:steps} applies,
and $\ctx(\pi'_1) = \ctx(\pi^0_{\ell^0+1}) = \pi_0$,
where $\pi'_1$ is the successor of $\pi_1$ in $\bar{\pi}$.

Now consider chains $\ochain {c_p} {s_p} {c_{p+1}}$ with $p \in [1:\ell]$,
corresponding to the portion $\pi_p \bar{\pi}_p \pi_{p+1}$ of the support.
If $s_p = \epsilon$, $\ctx(\pi_{p+1}) = \pi_0$ as in the base case.
Otherwise, let
$\bar{\pi}_p = \bar{pi}^p_0 \pi^p_1 \dots \pi^p_{\ell^p} \bar{\pi}^p_{\ell^p} \pi^p_{\ell^p+1}$.
By claim (b) of the induction hypothesis we have $\ctx(\stack(\pi^p_{\ell^p+1})) = \ctx(\pi'_p)$,
where $\pi'_p$ is the successor of $\pi_p$ in $\bar{\pi}$.
Since $\pi^p_{\ell^p+1}$ is a pop node, the fourth row of \cref{table:steps} applies,
and $\ctx(\pi_{p+1}) = \ctx(\stack(\pi^p_{\ell^p+1})) = \ctx(\pi'_p)$.
Since we already proved that $\ctx(\pi'_1) = \pi_0$,
by a simple induction we have $\ctx(\pi_p) = \pi_0$ for all $p \in [1:\ell+1]$,
and claim (a) is proved.

As for claim (b), from the inductive hypothesis applied to $s_0$ follows that
the pop node $\pi^0_{\ell^0+1}$ that precedes $\pi_1$ is such that
$\ctx(\stack(\pi^0_{\ell^0+1})) = \ctx(\pi'_0)$,
hence by the fourth row of \cref{table:steps} we have $\ctx(\pi_1) = \ctx(\pi'_0)$.
By the second and third rows of \cref{table:steps}
we have $\stack(\pi_p) = \pi_1$ for all $p \in [2:\ell+1]$,
so in particular $\ctx(\stack(\pi_{\ell+1})) = \ctx(\pi_1) = \ctx(\pi'_0)$.
\qed
\end{proof}

Hence, for any two positions $i,j$ in $w_{\bar{u}}$ such that $\chain(i,j)$,
there is a support
$\bar{\pi}_0 \pi_1 \bar{\pi}_1 \pi_2 \dots \pi_\ell \bar{\pi}_\ell \pi_{\ell+1}$
such that, naming $\pi_0$ the predecessor of $\bar{\pi}_0$ in $\bar{\pi}$,
and $\pi_{\ell+2}$ the successor of $\pi_{\ell+1}$,
we have $\iota(i) = \pi_0$ and $\iota(j) = \pi_{\ell+1}$,
and for all $p \in [2:\ell+1]$ we have $\ctx(\pi_p) = \pi_0$.

Conversely, for each pop node $\pi_p$, there exists a node $\pi_q$
which is the closest push or shift successor of $\pi_p$
such that $\chain(\iota^{-1}(\ctx(\pi_p)), \iota^{-1}(\pi_q))$.

We can now prove the soundness result.
\soundnessthm*
\begin{proof}
  The proof is carried out by structural induction on the formula. For each
  $\psi \in \clos{\phi}$ and position $i \in [1:k]$ in $w_{\bar{u}}$, if $\psi
  \in \Gamma(\iota(i))$ then $(w_{\bar{u}}, i) \models \psi$. It follows
  directly that $(w_{\bar{u}}, 1) \models \phi$.
  
  For the base case of a simple proposition, the claim follows directly from the
  definition of $w_{\bar{u}}$ whose $i$-th letter is exactly
  $\Gamma_{AP}(\iota(i))$. For the inductive case we branch over the kind of
  formula:
  \begin{enumerate}
    \item the case of Boolean connectives is trivial. For example, if
      $\psi_1\lor\psi_2\in\Gamma(\iota(i))$ then by the expansion rules
      (\cref{table:expansion}) we have either $\psi_1\in\Gamma(\iota(i))$ or
      $\psi_2\in\Gamma(\iota(i))$, hence by inductive hypothesis either
      $(w_{\bar{u}}, i) \models \psi_1$ or $(w_{\bar{u}}, i) \models \psi_2$,
      therefore $(w_{\bar{u}}, i) \models \psi_1\lor\psi_2$.
    \item if $\lnextsup{t} \psi\in\Gamma(\iota(i))$, then by the step rules of  
      \cref{table:steps} (either the one for push or for shift nodes) we have
      $\psi\in\Gamma(\iota(i+1))$. Therefore by inductive hypothesis we have
      $(w_{\bar{u}}, i+1) \models \psi$. Furthermore we know
      \cref{rule:check:pnext} did not trigger because the branch is accepted, so
      we know that either $i \sim^t (i+1)$ or $i \doteq (i+1)$. Hence
      $(w_{\bar{u}}, i) \models\lnextsup{t} \psi$.
    \item if $\lcnext{t} \psi\in\Gamma(\iota(i))$. Since the word is accepted we
      know that the \emph{accepting rule} has fired, and since $\iota(i)$ is a
      push or shift node it means there is at least a pop node descendant
      $\iota(k)$ such that $\stack(\iota(k))=\iota(i)$ (or
      $\stack(\iota(k))=\stack(\iota(i))$ if $\iota(i)$ is a shift node). Since
      \cref{rule:check:xnext} did not fire on $\iota(k)$, we know that
      $\lcnext{t} \psi$ is \emph{fulfilled} in $\iota(i)$. Then,
      \cref{def:xnext-fulfill} tells us the existence of a pop node $\iota(j)$
      (possibly but not necessarily equal to $\iota(k)$) such that:
      \begin{enumerate}
        \item\label{soundness:ctx} 
          $\ctx(\iota(j))=\iota(i)$, 
        \item\label{soundness:gamma} 
          $\Gamma(\iota(i)) \lessdot \Gamma(\iota(j))$ or $\Gamma(\iota(i)) 
          \doteq \Gamma(\iota(j))$, and
        \item\label{soundness:z}
          $\alpha\in\Gamma(\iota(z))$ for some push or shift node $\iota(z)$
          descendant of $\iota(j)$. 
      \end{enumerate}
      By \cref{soundness:ctx,soundness:z}, together with
      \cref{lemma:soundness-help}, we know that $\chain(i,z)$ holds. This last
      fact and \cref{soundness:gamma} tell us that 
      $(w_{\bar{u}}, i)\models\lcnext{t} \psi$.
      \item The case of other operators follow similar arguments.
  \end{enumerate}
\end{proof}

\begin{figure}
  \centering
  \input{sections/tableau_example}
  \caption{Example of an accepting branch of the tableau for formula
    $\lcdnext \lcall \land \lcunext \lexc \land \lwcdnext p$.}
  \label{fig:tableau-example}
\end{figure}

\section{Appendix: Tableau Encoding}
\label{app:encoding}

Here we present the details of the tableau SMT encoding that are missing from
\cref{sec:tableau}. We start by noting that, to handle hierarchical operators as
well, the definition of $\xnf$ in \cref{def:encoding:xnf} has to be extended
with the following two inductive cases:
\begin{align*}
  \xnf(\lhuntil{t}{\alpha}{\beta}) & = 
    (\xnf(\beta) \land \zeta_t) \lor \big( 
      \xnf(\alpha) \land \lhnext{t}(\lhuntil{t}{\alpha}{\beta})
    \big) \\
  \xnf(\lhrelease{t}{\alpha}{\beta}) & = 
    (\zeta_t \implies \xnf(\beta)) \land \big( 
      \xnf(\alpha) \lor \lwhnext{t}(\lhrelease{t}{\alpha}{\beta})
    \big)
\end{align*}
where $\zeta_u$ and $\zeta_d$ are new auxiliary symbols from $\mathcal{S}$.

Then, we proceed with the definition of $\phi_{\mathrm{axioms}}$, $\phi_{OPM}$
and $\phi_{\mathrm{init}}$:
\begin{align*}
  \phi_{\mathrm{axioms}} \equiv  {} &
  \bigwedge_{p\in\Sigma} \bar\Sigma(p) \land 
  \bigwedge_{p\in AP}\neg\bar\Sigma(p) \\
  {}\land {} & \forall x (
      \bar\Sigma(\struct(x)) 
      \land \bar\Sigma(\smb(x))
      \land \Gamma(\struct(x),x)
  )\\
  \phi_{OPM} \equiv {} & 
  \bigwedge_{p\lessdot q} p \lessdot q \land
  \bigwedge_{p\doteq q} p\doteq q \land
  \bigwedge_{p\gtrdot q} p\gtrdot q \\
  {}\land{} & \bigwedge_{p\not\lessdot q} \neg(p \lessdot q) \land
  \bigwedge_{p\not\doteq q} \neg(p\doteq q) \land
  \bigwedge_{p\not\gtrdot q} \neg(p\gtrdot q)\\
  \phi_{\mathrm{init}} \equiv {} &
  \stack(0)=0\land\ctx(0)=0\land\struct(0)=\#\land\smb(0)=\#\\
    {}\land{} & \smb(1)=\#\land\stack(1)=0\land\ctx(1)=0 \\
    {}\land{} & \xnf(\phi)_G(1)
\end{align*}

Then, we present the missing formulas for step rules:
\begin{align*}
  \steprule_{\shift}(x)\equiv {} & 
  \bigwedge_{\lnextsup{t}\alpha\in\clos{\phi}}\big(
    \Gamma((\lnextsup{t}\alpha)_G, x) \implies \xnf(\alpha)_G(x+1)
  \big)\\
  {} \land {} & \smb(x+1)=\struct(x) \\
  {} \land {} & \stack(x+1)=\stack(x) \\
  {} \land {} & \ctx(x+1)=\ctx(x)\\
  \steprule_{\pop}(x)\equiv{} & 
  \forall p (\Gamma(p,x) \iff \Gamma(p,x+1))\\
  {} \land {} & \smb(x+1)=\smb(\stack(x)) \\
  {} \land {} & \stack(x+1)=\stack(\stack(x)) \\
  {} \land {} & \ctx(x+1)=\ctx(\stack(x))
\end{align*}

Finally, we present the termination rules. As mentioned in \cref{sec:encoding},
there is a formula $r_i(x)$ encoding each $i$-th block of rows in
\cref{table:rejecting}. As mentioned, \cref{rule:contradiction} does not need to
be encoded, and \cref{rule:conflict,rule:end} have already been presented. So we
present here the encoding formulas from $r_{\ref{rule:check:pnext}}(x)$ to
$r_{\ref{rule:prune}}(x)$.
\begin{align*}
  r_{\ref{rule:check:pnext}}(x)\equiv {} &
  (\push(x-1)\lor\shift(x-1))\implies{}\\
  &
  \left[
    \begin{aligned}
      &\smashoperator{\bigwedge_{\lnextsup{d}\alpha\in\clos{\phi}}} \neg \bigl(
        \Gamma((\lnextsup{d}\alpha)_G, x-1) \land \struct(x-1)\gtrdot\struct(x) 
      \bigr) \\
      {}\land{}&
      \smashoperator{\bigwedge_{\lnextsup{u}\alpha\in\clos{\phi}}} \neg \bigl(
        \Gamma((\lnextsup{u}\alpha)_G, x-1) \land \struct(x-1)\lessdot\struct(x) 
      \bigr)
    \end{aligned}
  \right]\\
  r_{\ref{rule:check:wpnext}}(x)\equiv {} & 
  (\push(x-1)\lor\shift(x-1))\implies{}\\
  &
  \left[
    \begin{aligned}
      &\smashoperator{\bigwedge_{\lwdnext \alpha \in \clos{\phi}}} \bigl(
        \Gamma((\lwdnext \alpha)_G, x-1) \land \neg(\struct(x-1)\gtrdot\struct(x))
      \bigr)
      \implies \xnf(\alpha)_G(x)
      \\
      {}\land{}&
      \smashoperator{\bigwedge_{\lwunext \alpha\in\clos{\phi}}} \bigl(
        \Gamma((\lwunext \alpha)_G, x-1) \land \neg(\struct(x-1)\lessdot\struct(x))
      \bigr)
      \implies \xnf(\alpha)_G(x)
    \end{aligned}
  \right]
  \shortintertext{then:}
  r_{\ref{rule:check:xnext}}(x) \equiv {}&
  \forall y \left(
    \begin{aligned}
    &y < x \land (y = \stack(x) \lor \stack(y) = \stack(x)) \implies \\
    &\smashoperator[r]{\bigwedge_{\lcnext{d/u} \alpha\in\clos{\phi}}}
    \bigl(
    \Gamma((\lcnext{d/u} \alpha)_G, y)
      \implies \mathrm{satisfied}((\lcnext{d/u} \alpha)_G, x, y)
    \bigr)
    \end{aligned}
  \right)\\
  r_{\ref{rule:check:wxnext}}\equiv{} & 
  \smashoperator[l]{\bigwedge_{\lwcnext{d/u} \alpha\in\clos{\phi}}}
  \left(
    \Gamma((\lwcnext{d/u} \alpha)_G, \ctx(x)) \land \dirpr_{d/u}(\ctx(x), x)
    \implies \xnf(\alpha)_G(x)
  \right)
  \shortintertext{where:}
  &\hspace*{-3em}\begin{aligned}
    \mathrm{satisfied}((\lcnext{d/u} \alpha)_G, x, y) \equiv {}&
    \exists z \left(
      \begin{aligned}
      &y \le z \le x \land \pop(z) \land \ctx(z) = y \land {} \\
      &\dirpr_{d/u}(y, z) \land \xnf(\alpha)_G(z)
      \end{aligned}
    \right)\\
   \dirpr_d(x, y) \equiv {} &
     \struct(x) \lessdot \struct(y) \lor \struct(x) \doteq \struct(y) \\
   \dirpr_u(x, y) \equiv {} &
     \struct(x) \gtrdot \struct(y) \lor \struct(x) \doteq \struct(y)
  \end{aligned}
\end{align*}
Now, for a family of formulas $\mathcal{F}$, we define the following auxiliary
formulas:
\begin{align*}
  \any_{\mathcal{F}}(x)\equiv{}& \bigvee_{\psi\in\mathcal{F}} \Gamma((\psi)_G, x)\\
  \hucond_{\mathcal{F}}(x) \equiv {} &
    \any_{\mathcal{F}}(x) \implies
      \left(\pop(x) \lor \left(\push(x) \land \pop(x-1)\right)\right) \\
  \hdcond_{\mathcal{F}}(x) \equiv {} &
    (\neg\pop(x) \land \any_{\mathcal{F}}(x)) \implies \push(x+1)\\
    {}\land{}&
    (\pop(x) \land \any_{\mathcal{F}}(\ctx(x))) \implies \neg\shift(x+1)
\end{align*}
with these in place, we can continue:
\begin{align*}
  r_{\ref{rule:check:hnextu}}(x) \equiv {}&
   \hucond_{\set{\lhunext \alpha\in\clos{\phi}}}(x-1) \\
    {}\land{}&\pop(x-1) \implies
      \smashoperator{\bigwedge_{\lhunext \alpha\in\clos{\phi}}}
      \left(
        \Gamma((\lhunext \alpha)_G, \stack(x-1))
          \implies \xnf(\alpha)_G(x-1) \land \push(x)
      \right) \\
  r_{\ref{rule:check:whnextu}}(x) \equiv {} &
    (\pop(x-1) \land \push(x) \land \pop(\stack(x-1)-1) \land \push(\stack(x-1))) \\
    {} \implies {} &
    \smashoperator{\bigwedge_{{\lwhunext \alpha\in\clos{\phi}}}}
      \Gamma((\lwhunext \alpha)_G, \stack(x-1)) \implies \xnf(\alpha)_G(x-1) \\
  r_{\ref{rule:check:hnextd}}(x) \equiv {}&
    \hdcond_{\set{\lhdnext \alpha\in\clos{\phi}}}(x)\\
    {}\land{}&\left(
    \begin{aligned}
      &\pop(x-1) \land \pop(x) \implies {} \\
    &\smashoperator[r]{\bigwedge_{\lhdnext \alpha\in\clos{\phi}}}
      \left(
      \Gamma((\lhdnext \alpha)_G, \ctx(x-1))
         \implies \xnf(\alpha)_G(\ctx(x-2)) \land \pop(x-2)
      \right) 
    \end{aligned}\right)\\
  r_{\ref{rule:check:whnextd}}(x) \equiv {} &
    (\pop(x-2) \land \pop(x-1) \land \pop(x)) \\
    {}\implies{} &
    \smashoperator{\bigwedge_{\lhdnext \alpha\in\clos{\phi}}}
      \Gamma((\lhdnext \alpha)_G, \ctx(x-1)) \implies
        \xnf(\alpha)_G(\ctx(x-2))\\
  r_{\ref{rule:check:huntild}}(x) \equiv {} &
  \hdcond_{\set{\zeta_d}}(x-1) \land
    \big(\pop(x-1) \land \pop(x) \implies \Gamma(\zeta_d, x-1)\big)
\shortintertext{and finally:}
  r_{\ref{rule:prune}} \equiv {} &
    \exists y \left(
      \begin{aligned}
        &\pending(x, y) \land {} \\
        &\forall p (\Gamma(p,x)\iff\Gamma(p,y)) \land {}\\
        & \smb(x)=\smb(y) \land {}\\
        &\forall p (\Gamma(p,\stack(x))\iff\Gamma(p,\stack(y))) \land {}\\
        &\forall p (\Gamma(p,\ctx(x))\iff\Gamma(p,\ctx(y)))
      \end{aligned}
    \right)
\shortintertext{where:}
  \pending(k, x) \equiv {} &
  (\push(x) \land \neg\exists y (y < k \land \pop(y) \land \stack(y)=x)) \\
  {}\lor{} &
  (\shift(x) \land \neg\exists y (y < k \land \pop(y) \land \stack(y)=\stack(x)))
\end{align*}

\section{Appendix: Program Encoding}
\label{app:program-encoding}

We model procedural programs in MiniProc, the programming language used in \cite{ChiariMPP23},
and we use the same approach to translate them to \emph{extended OPA}.
In \cite{ChiariMPP23} extended OPA are then translated to normal OPA following an established approach
for modeling procedural programs by means of pushdown automata~\cite{AlurBE18}.
We instead encode them directly in SMT.

An \emph{extended OPA} is a tuple
$\mathcal{M} = (\Sigma, \allowbreak M, AP, \allowbreak L, \allowbreak I, \mathcal{F}, V_\mathit{glob}, (V_f)_{f \in \mathcal{F}}, \allowbreak \delta)$
where $(\Sigma, \allowbreak M)$ is an OP alphabet,
$L$ is a finite set of locations,
$I \subseteq L$ is a set of initial locations,
$\mathcal{F}$ is a finite set of function names,
$V_\mathit{glob}$ is a finite set of global program variable names,
$(V_f)_{f \in \mathcal{F}}$ is a collection of finite sets of local variable names (which we assume disjoint among them and with $V_\mathit{glob}$);
$\delta$ is a triple of transition relations
\begin{gather*}
\delta_\mathit{push}, \delta_\mathit{shift} \subseteq L \times (\powset{AP \cup \Sigma} \times \mathrm{Guards} \times \mathrm{Actions}) \times L, \\
\delta_\mathit{pop} \subseteq L \times (L \times \mathrm{Guards} \times \mathrm{Actions}) \times L.
\end{gather*}
The set of all variables is $\mathcal{V} \equiv V_\mathit{glob} \cup \bigcup_{f \in \mathcal{F}} V_f$.

$\mathrm{Guards}$ are Boolean combinations of comparisons between integer expressions involving variables, integer constants and arithmetic operations.
Variable types can be fixed-size signed or unsigned integers (Boolean variables are encoded as 1-bit integers),
and fixed-size arrays thereof.
Integer variables are represented with the theory of fixed-size bitvectors,
and arrays with the theory of arrays.
For simplicity, we show the encoding only for scalar variables,
while the encoding of arrays can be defined analogously.

$\mathrm{Actions}$ can be of the following types:
\begin{itemize}
\item $(\mathrm{Assign}, v, E)$, where $v \in \mathcal{V}$ is the variable to be assigned and $E$ is a MiniProc expression;
\item $(\mathrm{Nondet}, v)$, where $v \in \mathcal{V}$ is the variable to be assigned a nondeterministic (symbolic) value;
\item $(\mathrm{Call}, f, (v_1, \dots, v_n), (E_1, \dots, E_n), (r_1, r_m), (t_1, t_m))$,
    where $f \in \mathcal{F}$, for $i \in \{1, \dots, n\}$ all $v_i \in V_f$ are formal parameters passed by value and $E_i$ are the actual parameters,
    and for $j \in \{1, \dots, m\}$ all $r_j \in V_f$ are the formal parameters passed by value-result,
    while the respective actual parameters are $t_j \in \mathcal{V}$;
\item $(\mathrm{Return}, f, (r_1, r_m), (t_1, t_m))$ where $f \in \mathcal{F}$ and $r_i$ and $t_i$ are as above;
\item $\mathrm{Noop}$.
\end{itemize}

The SMT encoding is as follows.
We define a sort $\mathcal{L}$ for locations in $L$ and an uninterpreted function
$\mathrm{pc} : \mathbb{N} \to \mathcal{L}$ that associates each step to a location of the extended OPA.
Variable values are encoded through a valuation function
$\val : V \times \mathbb{N} \to \mathbb{BV}$,
where $\mathbb{BV}$ is the set of fixed-size bit vectors.

Given an expression $E$,
by $E|_i$ we denote the expression where all variables $v \in V$
have been replaced with $\val(v, i)$.

We define the following predicates:
\[
\steprule^\delta_{\push}(x) \equiv
  \smashoperator{\bigvee_{\delta_\mathit{push}(l_1, (b, g, a), l_2)}}
  \big(
    \mathrm{pc}(x, l_1)
    \land \action(x, \cdot, a)
    \land \mathrm{label}(b, x)
    \land g|_x
    \land \mathrm{pc}(x+1, l_2)
  \big)
\]
Where
\(
\mathrm{label}(b, x) \equiv
  \bigwedge_{\mathrm{p} \in b} \Gamma(\mathrm{p}, x)
  \land
  \bigwedge_{\mathrm{p} \in b \setminus \mathcal{V}} \neg \Gamma(\mathrm{p}, x)
\),
and
\[
\steprule^\delta_{\pop}(x) \equiv
  \smashoperator[l]{\bigvee_{\delta_\mathit{pop}(l_1, (l_s, g, a), l_2)}}
  \begin{aligned}[t]
    &\big(\mathrm{pc}(x, l_1)
      \land \mathrm{pc}(\stack(x), l_s)
      \land \action(x, \stack(x), a) \\
      &\land g|_x
      \land \mathrm{pc}(x+1, l_2)
    \big)
    \end{aligned}
\]
and $\steprule^\delta_{\shift}(x)$, which is the same as $\steprule^\delta_{\push}(x)$,
but replacing $\delta_\mathit{push}$ with $\delta_\mathit{shift}$.

To propagate values of variables in a set $V$ from $x$ to the next time step,
we define an auxiliary predicate
\[
\prop(V, x) \equiv
\bigwedge_{v \in V} \val(v, x+1) = \val(v, x)
\]
The $\action$ predicates are defined as follows:
\begin{align*}
&\action(x, \cdot, (\mathrm{Assign}, v, E)) \equiv
    \val(v, x+1) = E|_x
    \land \prop(\mathcal{V} \setminus \{v\}, x) \\
&\action(x, \cdot, (\mathrm{Nondet}, v)) \equiv
  \prop(\mathcal{V} \setminus \{v\}, x) \\
&\action(x, \cdot, (\mathrm{Call}, f, (v_1, \dots, v_n), (E_1, \dots, E_n), (r_1, r_m), (t_1, t_m))) \equiv \\
  &\qquad
  \begin{aligned}
    &\bigwedge_{i = 1}^n \val(v_i, x+1) = E_i|_x \\
    &\land \bigwedge_{j = 1}^m \val(r_j, x+1) = \val(t_j, x) \\
    &\land \bigwedge_{s \in R} \val(s, x+1) = 0 \text{ where $R = V_f \setminus (\{v_1, \dots, v_n\} \cup \{r_1, \dots, r_m\})$} \\
    &\land \prop(\mathcal{V} \setminus V_f, x)
  \end{aligned} \\
&\action(x, y, (\mathrm{Return}, f, (r_1, r_m), (t_1, t_m))) \equiv \\
  &\qquad
  \begin{aligned}
    &\bigwedge_{j = 1}^m \val(t_j, x+1) = \val(r_j, x) \\
    &\land \bigwedge_{s \in R} \val(s, x+1) = \val(s, y) \text{ where $R = V_f \setminus \{t_1, \dots, t_m\}$} \\
    &\land \prop(\mathcal{V} \setminus (V_f \cup \{t_1, \dots, t_m\}), x)
  \end{aligned} \\
&\action(x, \cdot, \mathrm{Noop}) \equiv
  \prop(\mathcal{V}, x)
\end{align*}

We also define a predicate that transforms variable names into atomic propositions
depending on their value in a given time step:
\[
\varap(x) \equiv
\smashoperator{\bigwedge_{v \in \mathcal{V} \cap AP}} \Gamma(v, x) \iff \val(v, x) \neq 0
\]

Then, we can define $\unr{\mathcal{M}}_k$, the length-$k$ unravelling of $\mathcal{M}$, as follows:
\begin{align*}
  &\bigvee_{l \in I} \pc(0, l) \land
  \bigwedge_{v \in \mathcal{V}} \val(v, 0) = 0 \land {}\\
  & \quad\forall x (1 \leq x < k \implies (
  \begin{aligned}[t]
    & \varap(x) \land {}\\
    & (\push(x) \implies \steprule^\delta_{\push}(x)) \land {}\\
    & (\shift(x) \implies \steprule^\delta_{\shift}(x)) \land {}\\
    & (\pop(x) \implies \steprule^\delta_{\pop}(x))))
  \end{aligned}
\end{align*}


%% file: sections/tableau_example.tex
\tikzset{
  ctx/.style={->, dotted, blue, >=latex},
  stack/.style={->, dashed, red, >=latex},
  baseline=(current bounding box.north)
}
\begin{forest}
  for tree={
    myleaf/.style={label=below:{\strut#1}}
  },
  [{$\{\lcdnext \lcall \land \lcunext \lexc \land \lwcdnext p\}$,$\#$,$\bot$,$\bot$},
    name=root
    [{$\{\lcdnext \lcall, \lcunext \lexc, \lwcdnext p\}$,$\#$,$\bot$,$\bot$},
      edge label={node[midway,left,font=\scriptsize]{expansion}}
      [{\underline{1,$\{\lcall, \lcdnext \lcall, \lcunext \lexc, \lwcdnext p\}$,$\#$,$\bot$,$\bot$}},
        edge label={node[midway,left,font=\scriptsize]{add struct}},
        name=push1
        [{$\emptyset$,$\lcall$,1,$\bot$},
          edge label={node[midway,left,font=\scriptsize]{push}}
          [{\underline{2,$\{\lcall\}$,$\lcall$,1,$\bot$}},
            edge label={node[midway,left,font=\scriptsize]{add-struct}},
            name=push2
            [{$\emptyset$,$\lcall$,2,1},
              edge label={node[midway,left,font=\scriptsize]{push}}
              [{\underline{3,$\{\lret\}$,$\lcall$,2,1}},
                edge label={node[midway,left,font=\scriptsize]{add struct}},
                name=shift3
                [{$\emptyset$,$\lret$,2,1},
                  edge label={node[midway,left,font=\scriptsize]{shift}}
                  [{\underline{4,$\{\lcall\}$,$\lret$,2,1}},
                    edge label={node[midway,left,font=\scriptsize]{add-struct}},
                    name=pop4
                    [{$\{\lcall\}$,$\lcall$,1,$\bot$},
                      edge label={node[midway,left,font=\scriptsize]{pop}}
                      [\underline{5,$\{\lcall, p\}$,$\lcall$,1,$\bot$},
                        edge label={node[midway,left,font=\scriptsize]{guess}},
                        name=push5
                        [{$\emptyset$,$\lcall$,5,1},
                          edge label={node[midway,left,font=\scriptsize]{push}}
                          [\underline{6,$\{\lexc\}$,$\lcall$,5,1},
                            edge label={node[midway,left,font=\scriptsize]{add struct}},
                            name=pop6
                            [{$\{\lexc\}$,$\lcall$,1,$\bot$},
                              edge label={node[midway,left,font=\scriptsize]{pop}}
                              [\underline{7,$\{\lexc\}$,$\lcall$,1,$\bot$},
                                edge label={node[midway,left,font=\scriptsize]{guess}},
                                name=pop7
                                [\underline{8,$\{\lexc\}$,$\#$,$\bot$,$\bot$},
                                  edge label={node[midway,left,font=\scriptsize]{pop}},
                                  name=push8
                                  [{$\emptyset$,$\lexc$,8,$\bot$},
                                    edge label={node[midway,left,font=\scriptsize]{push}}
                                    [\underline{9,$\{\#\}$,$\lexc$,8,$\bot$},
                                      edge label={node[midway,left,font=\scriptsize]{add struct}},
                                      name=pop9
                                      [\underline{10,$\{\#\}$,$\#$,$\bot$,$\bot$},
                                        edge label={node[midway,left,font=\scriptsize]{pop}},
                                        myleaf={\scriptsize \cmark EMPTY}
                                      ]
                                    ]
                                  ]
                                ]
                              ]
                            ]
                          ]
                        ]
                      ]
                    ]
                  ]
                ]
              ]
            ]
          ]
        ]
      ]
    ]
  ]
\draw[ctx] (pop4.east) to[right, out=60, in=260] (push1.east);
\draw[ctx] (pop6.east) to[right, out=70, in=270] (push1.east);
%
\draw[stack] (push2.west) to[left, out=170, in=260] (push1.west);
\draw[stack] (shift3.west) to[bend left, left] (push2.west);
\draw[stack] (pop4.west) to[bend left, left] (push2.west);
\draw[stack] (push5.west) to[left, out=110, in=265] (push1.west);
\draw[stack] (pop6.west) to[bend left, left] (push5.west);
\draw[stack] (pop7.west) to[left, out=105, in=265] (push1.west);
\draw[stack] (pop9.west) to[bend left, left] (push8.west);
\end{forest}
